\documentclass[
  10pt,
  letterpaper,
  twocolumn,
  amsmath,
  amssymb,
  aps,
  pra,
  floatfix,
  superscriptaddress,
  longbibliography
]{revtex4-2}

\usepackage{amsmath, amssymb, amsthm}
\usepackage{physics}
\usepackage{graphicx}
\usepackage{hyperref}
\usepackage{tikz}
\usetikzlibrary{quantikz2}
\usetikzlibrary{arrows.meta,positioning,calc}
\usepackage{iftex}
\ifPDFTeX
  \usepackage[utf8]{inputenc}
  \usepackage[T1]{fontenc}
\fi
\usepackage{dcolumn}
\usepackage{bm}
\usepackage{braket}
\usepackage{microtype}
\usepackage[ruled, vlined]{algorithm2e}

\allowdisplaybreaks
\hypersetup{
  colorlinks=true,
  linkcolor=blue,
  citecolor=blue,
  urlcolor=blue,
  breaklinks=true
}
\newtheorem{theorem}{Theorem}

\newtheorem{definition}{Definition}
\newtheorem{proposition}[theorem]{Proposition}
\newtheorem{corollary}[theorem]{Corollary}
\newtheorem{remark}[theorem]{Remark}

\begin{document}
    \title{Congestion-free routing on quantum chips} 
    \author{Mithilesh Kumar}
    \thanks{\href{mailto:mithilesh.kumar@krea.edu.in}{\nolinkurl{mithilesh.kumar@krea.edu.in}}}
    \affiliation{Krea University, Sri City, India}
    
    \author{Yusuf Tahir}
    \thanks{\href{mailto:yusuf_tahir.sias25@krea.ac.in}{\nolinkurl{yusuf_tahir.sias25@krea.ac.in}}}
    \affiliation{Krea University, Sri City, India}

    \author{Varun Daiya}
    \thanks{\href{mailto:varun_daiya.sias25@krea.ac.in}{\nolinkurl{varun_daiya.sias25@krea.ac.in}}}
    \affiliation{Krea University, Sri City, India}

    \author{Sanjana Mattaparthi}
    \thanks{\href{mailto:sanjana_mattaparthi.sias22@krea.ac.in}{\nolinkurl{sanjana_mattaparthi.sias22@krea.ac.in}}}
    \affiliation{Krea University, Sri City, India}

    \author{Aarav Shaurya}
    \thanks{\href{mailto:aarav_shaurya.sias22@krea.ac.in}{\nolinkurl{aarav_shaurya.sias22@krea.ac.in}}}
    \affiliation{Krea University, Sri City, India}
    
\date{\today}

\begin{abstract}
Limited connectivity makes nonlocal quantum gates expensive on near-neighbor hardware, where compilation typically relies on SWAP transport, inheriting both depth overhead and path congestion. We present a swap-free routing framework in which higher levels of a qudit act as orthogonal spectral buses that transport control information without moving the computational state. We show that exact congestion relief in nearest-neighbor architectures requires local Hilbert-space expansion. In this model, a nonlocal operation over a path of length $L$ requires $2L+1$ logical routing primitives, compared to the $3L$ baseline. Overlapping routes remain distinguishable through bus labels encoded in the same physical qudits. This routing algebra extends to Boolean fan-in at a common target: multiple controls arriving on distinct buses trigger a local unitary based on an arbitrary Boolean function of bus digits, yielding multi-control operations of depth $2L + D_g + O(1)$ for fan-in size $K$ and target-synthesis cost $D_g$. We prove decodability, reversibility, and correctness for CNOT and Boolean fan-in, along with a state-count lower bound $d \geq 2^{K+1}$ for exact overlap routing. Cirq simulations confirm single-control correctness and zero crosstalk. Compiler-level benchmarks on QFT, QAOA, and mirror-interaction circuits verify the predicted congestion law and transport reduction. Noisy QuTiP simulations show that the architectural advantage depends on higher-level coherence and speed. These results identify spectral qudit routing as a congestion-relief architecture that separates nonlocal control delivery from local target-side aggregation, providing a minimal mechanism for overcoming qubit routing limitations.
\end{abstract}

\maketitle

\section{Introduction}
Quantum algorithms are usually written as though any pair of logical qubits can interact directly. Hardware is much less generous. In most current architectures, entangling gates are restricted to a sparse connectivity graph, so nonlocal interactions must be compiled into sequences of local operations. The standard remedy is SWAP transport: move the relevant state across the chip, apply the desired gate, and, when layout matters, undo the transport. 

As identified in seminal works by Li et al. (2019)\cite{li2019tackling} and Cowtan et al. (2019)\cite{cowtan2019qubit}, moving quantum states across a chip to facilitate gates introduces a significant depth overhead that scales with the distance between qubits. Furthermore, as processors scale, this movement creates severe path congestion, where overlapping routes must be serialized, leading to increased decoherence and error rates \cite{murali2019noise}.

In parallel, researchers explored moving beyond the binary qubit. Lanyon et al. (2009)\cite{lanyon2009simplifying} pioneered the use of higher-dimensional systems, or qudits, to simplify quantum logic, demonstrating that extra levels in trapped ions can replace multiple gates. This ``qudit'' approach was further refined by Gokhale et al. (2019)\cite{gokhale2019asymptotic} and summarized in the review by Wang et al. (2020)\cite{quditreview}, which frames high-dimensional computing as a method for optimizing quantum circuits.

One response to routing overhead is to replace state transport by signal propagation. Swapless constructions forward control information along a path and uncompute intermediate modifications, thereby reducing the need to shuttle the payload state itself \cite{swapless2024}. However, a fundamental bottleneck remains: when multiple nonlocal operations require the same spatial region, a qubit-only architecture cannot support their simultaneous passage without additional ancillas or temporal serialization.

This observation motivates a different use of multilevel quantum hardware. Rather than treating higher levels as passive encoding space, we use them as actively assigned routing resources. The central idea is to encode routed controls into orthogonal spectral buses of a single physical qudit. A node can then preserve its resident logical state in a two-level subspace while simultaneously carrying one or more routed bus labels in higher levels. In this way, congestion is shifted from space into local spectral structure: routing conflicts are resolved not by spatial separation, but by allocating distinct, locally addressable channels within the same physical system.

The present paper develops this idea as an exact nearest-neighbor routing model. A routed control is lifted onto a chosen bus, propagated along a path by reversible local updates, used to trigger a same-qudit target operation, and then removed by inverse cleanup. For a path of length $L$, this yields a routed schedule of $2L+1$ logical routing primitives. For comparison, the transport component of a naive SWAP-based implementation requires $L$ SWAPs, or $3L$ CNOT-equivalent layers. Beyond this depth reduction, the bus encoding ensures that intersecting routes remain algebraically separable, so congestion is addressed at the level of routing primitives as well as circuit depth. The graph-coloring structure underlying this behavior is confirmed by explicit compiled workloads studied below.

Our approach addresses the routing bottleneck by changing what is being propagated. Instead of moving computational states or relying on a single shared signal, control information is carried in distinct local channels within the same physical system. These channels are orthogonal subspaces (``spectral buses'') of a qudit, allowing multiple routed controls to pass through a node simultaneously without interference, provided sufficient local dimension is available. This makes the architectural claim explicit: spectral separation supplies concurrency without spatial duplication. We formalize this intuition by showing that, within the exact nearest-neighbor unitary model considered here, supporting overlapping routed controls while preserving local data requires an increase in local Hilbert-space dimension.

For a compiled layer of routed operations, the relevant complexity is governed not only by path length but also by the chromatic number of the induced route-conflict graph. In the exact-overlap model studied here, the limiting resource at a hotspot is local label capacity: parallel routes can coexist only if the shared node provides enough distinguishable channels. This yields a qualitative separation between qubit and qudit routing. Even the smallest hotspot with two overlapping routes cannot be realized in a one-qubit-per-node architecture that preserves resident data, yet it fits within a single logical round once a qudit can host two buses.

The same mechanism naturally extends to \emph{routed fan-in}. Once multiple controls have been delivered to a common target on distinct buses, the remaining task reduces to a local Boolean-controlled operation on the lifted target subspace. This separates distributed delivery from local aggregation and provides a useful structural decomposition of nonlocal operations.

\paragraph{Positioning}
Conceptually, the proposal differs from two common uses of higher-dimensional hardware. First, higher levels are not treated as passive encoding space but as actively allocated routing resources. Second, unlike single-signal swapless propagation schemes, overlapping routes remain distinct through explicit bus labeling that is locally decodable. The claim is therefore architectural: spectral separation supplies concurrency without spatial duplication.

\paragraph{Results}
At the level of logical routing primitives, the framework reduces the transport cost of a nonlocal controlled operation over a path of length $L$ from $3L$ layers (SWAP-based routing) to $2L+1$ logical routing steps. More importantly, it changes how concurrency scales: instead of forcing serialization at congested regions, it allows multiple routes to coexist at a node via distinct local channels. This leads to a congestion model governed by local channel capacity rather than graph distance alone. The formalism extends naturally to routed multi-control operations, separating nonlocal transport from local aggregation.

\paragraph{Organization}
The paper is organized as follows. Section~II introduces the spectral-routing formalism, including the same-qudit interpretation, routing algebra, and primitive operations. Section~III develops the routed nonlocal CNOT protocol and its congestion properties. Section~IV presents routed Boolean fan-in and numerical validation. Section~V discusses physical implementation, compilation, and limitations. Detailed proofs and extensions appear in the appendices.

\section{Spectral Routing Formalism}\label{sec:formalism}

We now introduce the idealized model used throughout the paper. The formalism is defined at the level of logical routing primitives: it specifies how routed information is encoded, propagated, applied locally at a target, and removed, while leaving hardware-level pulse decompositions to a later compilation stage.

\subsection{Qudit Architecture and State Space}

Let the quantum processor be represented by an undirected graph $G = (V, E)$, where each vertex $v \in V$ corresponds to a physical $d$-level quantum system (qudit). The Hilbert space at each node is $\mathcal{H}_d$, spanned by the orthonormal basis $\{|0\rangle, |1\rangle, \dots, |d-1\rangle\}$. We partition this space into a computational subspace and a routing subspace. The computational subspace is $\mathcal{H}_C=\text{span}\{|0\rangle, |1\rangle\}$, and logical quantum information is encoded exclusively in these two levels. The routing subspace is $\mathcal{H}_R=\text{span}\{|2\rangle, \dots, |d-1\rangle\}$, and these higher levels provide the auxiliary channels used for routing across the graph.

\subsection{Parallel Bus Abstraction}

We define a family of routing buses indexed by $k \in \{1,2,\dots, K\}$. Each bus corresponds to an orthogonal excitation offset in the qudit basis, and to ensure non-interference between buses we assign $\Delta_k = 2^k$. When bus $k$ carries a logical value $x \in \{0,1\}$, the physical qudit occupies the pair of levels $|\Delta_k\rangle$ and $|\Delta_k + 1\rangle$.
In particular, the first routing bus corresponds to the levels $|2\rangle$ and $|3\rangle$, matching the intended shift of the computational states $|0\rangle$ and $|1\rangle$ into the routing subspace.

If several buses are active simultaneously, the resulting qudit value takes the form
\[
    x_0 + \sum_{k=1}^{K} x_k \cdot\Delta_k \qquad x_0, x_k \in \{0,1\}
\]
Here $x_0$ is the logical bit stored in the least-significant position, while the higher binary digits encode the active buses. Because binary decompositions are unique, this aggregate value can be decoded without ambiguity. The largest simultaneous routing value is $1 + \sum_{k=1}^{K} 2^k = 2^{K+1} - 1$, so supporting $K$ parallel buses requires $d \geq 2^{K+1}$. Equivalently, $K \leq \lfloor \log_2 d \rfloor - 1$.
Thus, under this binary bus encoding, the available routing parallelism grows only logarithmically with the accessible qudit dimension $d$.

\paragraph{No-Aliasing Assumption}
All results in this paper assume $d > 2^{K+1} - 1$, equivalently $d \geq 2^{K+1}$ for integer $d$. This ensures that the aggregate routing value $x_0 + \sum_{k=1}^{K} x_k \cdot\Delta_k$ never wraps modulo $d$, so the extraction maps $f_k(r)$ remain unambiguous and no modular aliasing occurs.

Each physical qudit therefore stores two kinds of information at once: its own logical bit in the least-significant position and, when present, routed control signals in the higher binary positions. In this sense, the same local device acts simultaneously as data storage and as a multiplexed routing node.

The choice $\Delta_k = 2^k$ is deliberate. Binary offsets guarantee unique decoding when several routes overlap on the same physical qudit, so each active bus can be recovered by a simple bit-extraction rule. Alternative encodings, such as unary or more general modular labelings, could trade larger bus count for more complicated decoding and composition rules. The present choice therefore prioritizes exact separability and algebraic composability over maximizing the raw number of buses supported at fixed $d$.

If one forbids same-node overlap entirely, then a unary-style labeling can assign a dedicated level pair to each route and thereby support more non-overlapping bus labels at fixed $d$. That is, however, a different routing model: once several buses must coexist on the same physical qudit, the encoding must distinguish every subset of active buses together with the local logical bit. In that exact-overlap setting, unary labeling does not avoid the underlying state-count requirement; it merely shifts the burden into additional overlap states and more cumbersome decoding rules. By contrast, the binary encoding used here yields a particularly simple and composable algebra for routing, target control, and cleanup while directly representing simultaneous overlaps on a single qudit.

This encoding allows multiple routing processes to coexist via distinct level shifts, preventing collisions even when paths overlap.

\subsection{Routing Algebra}

The binary bus encoding above can be packaged into a compact algebraic object that isolates the essential structure used by the routing protocol.

\begin{definition}[Routing Algebra]
Fix offsets $\Delta_k = 2^k$ for $k=1,\dots,K$, and define the admissible routing-value set
\[
\mathcal{R}_K=\left\{x_0 + \sum_{k=1}^{K} x_k \cdot\Delta_k\ \middle|\
x_0, x_k \in \{0,1\}
\right\}
\]
Together with the extraction maps
\[
\begin{aligned}
f_0(r) &= r \bmod 2,\\
f_k(r) &= \left\lfloor \frac{r}{2^k} \right\rfloor \bmod 2
\qquad (k=1,\dots,K)
\end{aligned}
\]
we call the pair \(\bigl(\mathcal{R}_K,\{f_k\}_{k=0}^{K}\bigr)\) the routing algebra.
\end{definition}

\begin{proposition}[Uniqueness of Decoding]
Under the no-aliasing assumption $d > 2^{K+1}-1$, the map
\[
(x_0,x_1,\dots,x_K)\longmapsto x_0 + \sum_{k=1}^{K} x_k\cdot\Delta_k
\]
from $\{0,1\}^{K+1}$ into $\mathcal{R}_K$ is injective. Equivalently, every $r \in \mathcal{R}_K$ admits a unique decomposition
\[
r = x_0 + \sum_{k=1}^{K} x_k \cdot\Delta_k,
\qquad x_0,x_k \in \{0,1\}
\]
and the extraction maps recover the coefficients exactly, $f_k(r)=x_k$ for $k=0,1,\dots,K$.
\end{proposition}

\begin{proof}
Because $\Delta_k = 2^k$, each element of $\mathcal{R}_K$ is precisely the binary expansion of the coefficient tuple $(x_0,\dots,x_K)$. Binary expansions are unique, and the no-aliasing assumption ensures that no two such integers are identified modulo $d$. The formulas for $f_k$ are exactly the corresponding binary-digit extraction rules.
\end{proof}

\begin{proposition}[Composability of Bus Labels]
Let
\[
    r = x_0 + \sum_{k=1}^{K} x_k \cdot\Delta_k \in \mathcal{R}_K
\]
and let $B \subseteq \{1,\dots,K\}$ be a set of buses such that $x_k=0$ for every $k\in B$. Then
\[
    T_B(r) := r + \sum_{k\in B}\Delta_k
\]
also lies in $\mathcal{R}_K$, and its extracted digits satisfy
\[
f_j(T_B(r)) =
\begin{cases}
1, & j\in B,\\[4pt]
f_j(r), & j\notin B
\end{cases}
\]
Thus disjoint bus labels compose by bitwise union without disturbing the previously active buses.
\end{proposition}

\begin{proof}
If $x_k=0$ for every $k\in B$, then adding $\sum_{k\in B}\Delta_k$ simply changes those binary digits from $0$ to $1$ and leaves all other digits unchanged. Since the resulting coefficient tuple is again binary-valued, the new integer remains in $\mathcal{R}_K$, and the stated extraction identities follow immediately.
\end{proof}

\begin{proposition}[Reversibility of Routing Updates]
Suppose
\[
r = x_0 + \sum_{k=1}^{K} x_k \cdot\Delta_k \in \mathcal{R}_K
\]
and let $B \subseteq \{1,\dots,K\}$ satisfy $x_k=1$ for every $k\in B$. Then
\[
T_B^{-1}(r) := r - \sum_{k\in B}\Delta_k
\]
also lies in $\mathcal{R}_K$, and
\[
T_B^{-1}(T_B(r')) = r'
\]
for every $r' \in \mathcal{R}_K$ with $f_k(r')=0$ for all $k\in B$. In particular, controlled additions and controlled subtractions of routing offsets are inverse permutations on the admissible routing-value set.
\end{proposition}

\begin{proof}
Subtracting $\sum_{k\in B}\Delta_k$ from a value that already carries those buses changes the corresponding binary digits from $1$ to $0$ and leaves all other digits unchanged. Hence the result remains in $\mathcal{R}_K$. The inverse identity follows because the two updates undo the same set of digit flips.
\end{proof}

\begin{proposition}[State-Count Lower Bound for Exact Overlap Encoding]\label{prop:state-count-lower-bound}
Consider any exact same-qudit routing scheme that stores one local logical bit and $K$ independently active routing bits on a single $d$-level system, with the property that every tuple $(x_0,x_1,\dots,x_K) \in \{0,1\}^{K+1}$ is represented by a perfectly distinguishable physical basis state and that each routing bit can be recovered without ambiguity. Then necessarily $d \geq 2^{K+1}$. In particular, the binary routing algebra above saturates this lower bound.
\end{proposition}

\begin{proof}
There are $2^{K+1}$ possible assignments of the local bit together with the $K$ routing bits. Exact distinguishability and unambiguous recovery require distinct orthogonal states for distinct tuples, so the physical Hilbert space must contain at least $2^{K+1}$ basis states. The binary encoding used here realizes exactly one admissible routing value for each such tuple, and therefore attains this minimum.
\end{proof}

\begin{theorem}[Single-Qubit Impossibility for Exact Overlap Routing]\label{thm:single-qubit-impossibility}
Consider any exact same-node routing scheme that attempts to store, in one local physical system, both an arbitrary resident qubit state $\ket{\psi}\in\mathcal{H}_2$ and $K$ independently recoverable routed control bits. If the storage is required to be reversible and uses no additional colocated ancilla at that node, then the local physical Hilbert space must have dimension at least
\(2^{K+1}\). In particular, a single qubit cannot simultaneously encode one arbitrary resident state and even one independently recoverable routed control bit; therefore it cannot encode two independently recoverable routed control bits.
\end{theorem}

\begin{proof}
Exact reversible storage of an arbitrary resident qubit together with $K$ classical routed bits defines an isometric encoding
\[
V:\mathcal{H}_2\otimes \mathbb{C}^{2^K}\longrightarrow \mathcal{H}_d
\]
on the encoded local subspace. The domain of this isometry has dimension
\[
\dim(\mathcal{H}_2\otimes \mathbb{C}^{2^K}) = 2^{K+1}
\]
An isometry cannot map a larger-dimensional space into a smaller-dimensional one, so necessarily
\[
    d \geq 2^{K+1}
\]
For a single qubit, $d=2$, which violates this bound for every $K\geq 1$.
\end{proof}

\begin{remark}
Theorem~\ref{thm:single-qubit-impossibility} is stronger than the bit-level state-count bound above. It does not merely say that a qubit cannot store several classical routing labels; it says that a qubit cannot preserve an \emph{arbitrary quantum resident state} while also carrying even one exact independently recoverable routed control label in the same local system.
\end{remark}

\subsection{Same-Qudit Interpretation}

The computational bit and the routing labels are stored in the same physical qudit, not in separate physical registers. A basis value of the form $r = x_0 + \sum_{k=1}^{K} x_k \cdot\Delta_k$, with $x_0, x_k \in \{0,1\}$, simultaneously records the local logical value $x_0$ in the least-significant position and the active routing buses $x_k$ in the higher binary digits. For a fixed bus pattern $s = \sum_{k=1}^{K} x_k \cdot\Delta_k$, the pair of levels $|s\rangle$ and $|s+1\rangle$ forms a two-level block carrying the local logical value $x_0$.

Thus states such as $|2\rangle,|3\rangle$ or $|6\rangle,|7\rangle$ should be interpreted as the same qudit viewed inside different lifted routing blocks. This is the sense in which the protocol uses ``the same qudit at higher levels'' rather than introducing separate routing wires.

\subsection{Allowed Operations}

We assume five ideal operations. First, a local computational gate is any unitary $U \in SU(2)$ acting on $\mathcal{H}_C$. For a state $|\psi\rangle = \alpha|0\rangle + \beta|1\rangle$, the action is simply $U|\psi\rangle$, while the routing subspace $\mathcal{H}_R$ is left unchanged. More generally, once a qudit has been lifted to a block $\text{span}\{|s\rangle,|s+1\rangle\}$, the same $2 \times 2$ action can be embedded on that block while preserving the higher-level bus label $s$.
Second, for every edge $(i,j) \in E$, we assume an ordinary nearest-neighbor CNOT gate,
\[
    CX_{(i,j)} |x\rangle_i |y\rangle_j = |x\rangle_i |y \oplus x\rangle_j
\]
for $x,y \in \{0,1\}$. Third, we use the cyclic bus-shift operator
\begin{equation}\label{eqn:busshift}
    S_{\Delta} |x\rangle = |(x + \Delta) \bmod d\rangle
\end{equation}
Fourth, for $(i,j) \in E$, the controlled-bus-load operator writes the selected routing offset onto the target when the source computational bit equals $1$:
\begin{equation}\label{eqn:busload}
    CBL_{(i,j)}^{\Delta_k} |x\rangle_i |y\rangle_j = |x\rangle_i |(y + x \cdot \Delta_k) \bmod d\rangle_j
\end{equation}
for $x \in \{0,1\}$. Fifth, after the initial lift, propagation along the rest of the path is controlled not by a bare computational bit but by whether the previous qudit already carries bus $k$. For $(i,j) \in E$, we therefore define
\begin{equation}\label{eqn:buspropagate}
    BCP_{(i,j)}^{\Delta_k} |r\rangle_i |y\rangle_j=|r\rangle_i |(y + f_k(r)\cdot \Delta_k) \bmod d\rangle_j
\end{equation}
where
\[
    f_k(r) = \left\lfloor \frac{r}{2^k} \right\rfloor \pmod 2
\]

\paragraph{Interpretation}

Mathematically, the propagation operator $BCP^{\Delta_k}$ functions as a generalized, bus-resolved CNOT gate. Whereas a standard computational CNOT applies $y \mapsto y + x \pmod 2$, this operator first isolates the relevant control variable through the modulo-extraction rule $f_k(r)$ and then applies the targeted shift $\Delta_k$ inside the higher-dimensional qudit space. The control action is therefore strictly conditional: the shift propagates only when the preceding node actively carries the bus-$k$ component. In this sense, the operator $CBL^{\Delta_k}$ performs the initial conditional lift from the original control qubit onto bus $k$, while $BCP^{\Delta_k}$ implements the source-faithful propagation rule discussed in the transcript.

% \paragraph{Remark}

% Both $CBL^{\Delta_k}$ and $BCP^{\Delta_k}$ are unitary operators, as each acts by a controlled permutation of the qudit basis. In particular, $BCP^{\Delta_k}$ does not introduce a new physical register; it simply inspects whether the relevant binary routing component is present in the source qudit.

\subsection{Routing-Controlled Target Operation and Phase Flips}

To complete the non-local operation, we define a unitary that acts directly on the same lifted qudit at the target node.

Let
\[
\mathcal{S}_K=\left\{\sum_{j=1}^{K} x_j \cdot\Delta_j\ \middle|\ x_j \in \{0,1\}\right\}
\]
denote the set of possible routing offsets. For each $s \in \mathcal{S}_K$, the pair $|s\rangle, |s+1\rangle$ forms the two-level block that stores the local logical bit.

To recover whether bus $k$ is active, we use $f_k(s) = \left\lfloor \frac{s}{2^k} \right\rfloor \pmod 2$.

\paragraph{Generalized Controlled Unitary}

Let $U$ be any single-qubit unitary. We define the same-qudit routing-controlled operation
\begin{equation}\label{eqn:controlledunitary}
    CU_R^{\Delta_k}(U)=\sum_{s \in \mathcal{S}_K}\sum_{a,b \in \{0,1\}}
\left(U^{f_k(s)}\right)_{ab}\dyad{s+a}{s+b} + I_{\mathcal{L}_K^\perp}
\end{equation}
where
\[
\mathcal{L}_K
=
\mathrm{span}\{|s\rangle,|s+1\rangle : s \in \mathcal{S}_K\}.
\]
On the admissible lifted routing subspace this is a direct sum of the relevant two-level blocks, while outside that subspace it acts trivially.
The routing-controlled NOT used in the non-local CNOT protocol is recovered as the special case $CXR^{\Delta_k} := CU_R^{\Delta_k}(X)$.
Here $X$ denotes the Pauli-$X$ (bit-flip) gate acting on the relevant two-level computational block.

\paragraph{Key Property}

For any value of the form
\[
    s + x_0,
    \qquad
    s = \sum_{j=1}^{K} x_j \cdot\Delta_j \in \mathcal{S}_K,
    \qquad x_0 \in \{0,1\}
\]
we have
\[
    CU_R^{\Delta_k}(U) |s + x_0\rangle=\sum_{b \in \{0,1\}}\left(U^{x_k}\right)_{b x_0}|s+b\rangle
\]
where $x_k = f_k(s)$. Thus the local logical bit is acted on by $U$ if and only if bus $k$ is active, while the higher-level routing offset $s$ is preserved.

In particular, for the non-local CNOT target gate, $CXR^{\Delta_k}|s + x_0\rangle = |s + (x_0 \oplus x_k)\rangle$.

\paragraph{Interpretation}

The operator $CU_R^{\Delta_k}(U)$ does not split the target into separate route and data registers. Instead, it views the same qudit as a collection of two-level blocks indexed by routing offsets and applies $U$ to the least-significant logical bit precisely on those blocks whose $k$-th bus bit is active. Choosing $U = Z$ yields a swap-free non-local phase flip, while more general phase rotations arise from other diagonal choices of $U$.

\paragraph{Remark}

Since the control action depends only on the extracted binary digit $f_k(s)$, overlapping buses remain exactly separable. This is the algebraic reason why a single qudit can simultaneously preserve its own logical bit and carry several routing labels without mixing them. Because the operator is a direct sum of one-qubit blocks together with the identity on $\mathcal{L}_K^\perp$, it is unitary on the full local Hilbert space.

\subsection{Idealized Tensor-Factor Extension}

The same-qudit description above is the physical picture motivated by the transcript and board discussion: a single qudit is lifted into higher levels while still carrying its local logical value in the least-significant position. For algebraic manipulations, however, it is often convenient to rewrite that same structure as an idealized tensor-factor model. This is an abstract relabeling of the same lifted blocks, not a claim that the hardware contains separate routing and computational registers.

Define the routing-label space
\[
\mathcal{H}_B
=
\text{span}\{\, |s\rangle_B : s \in \mathcal{S}_K \,\}
\]
where $\mathcal{S}_K$ is the set of admissible routing offsets introduced above, and let $\mathcal{H}_2 = \text{span}\{|0\rangle_C, |1\rangle_C\}$ denote an abstract two-level logical space. The lifted same-qudit block decomposition induces a canonical identification
\[
\Phi : \mathcal{H}_d \supset \text{span}\{|s\rangle, |s+1\rangle : s \in \mathcal{S}_K\}
\longrightarrow
\mathcal{H}_B \otimes \mathcal{H}_2
\]
given on basis states by
\[
\Phi(|s + x_0\rangle) = |s\rangle_B \otimes |x_0\rangle_C,
\qquad s \in \mathcal{S}_K,\ x_0 \in \{0,1\}
\]

Under this identification, the routing offset is treated as an abstract routing label and the least-significant bit is treated as an abstract logical qubit. The physical same-qudit state is therefore recovered by $|s + x_0\rangle = \Phi^{-1} \left( |s\rangle_B \otimes |x_0\rangle_C \right)$.

\paragraph{Projector Form}

For each bus $k$, define the projectors on the routing-label space
\[
\Pi_k^{(1)} = \sum_{s \in \mathcal{S}_K} f_k(s)\dyad{s}_B,
\qquad
\Pi_k^{(0)} = I_{\mathcal{H}_B} - \Pi_k^{(1)}
\]
These projectors isolate whether bus $k$ is active while ignoring the local logical value carried by the same physical qudit.

\paragraph{Tensor-Factor Controlled Unitary}

Let $U$ be any single-qubit unitary on $\mathcal{H}_2$. The routing-controlled operation can then be written in the abstract tensor-factor form
\[
    \widetilde{CU}_R^{\Delta_k}(U)=\Pi_k^{(0)} \otimes I_{\mathcal{H}_2}+\Pi_k^{(1)} \otimes U=\sum_{s \in \mathcal{S}_K} \dyad{s}_B \otimes U^{f_k(s)}
\]
This is precisely the higher-level projector formalism that treats routing labels and logical data as separate algebraic factors.

\paragraph{Equivalence with the Same-Qudit Picture}

The same-qudit operator defined earlier and the tensor-factor operator above are equivalent descriptions of the same action on the lifted routing subspace:
\[
    CU_R^{\Delta_k}(U)|_{\mathcal{L}_K}=\Phi^{-1}\,\widetilde{CU}_R^{\Delta_k}(U)\,\Phi
\]
Equivalently, for every lifted basis state,
\[
\widetilde{CU}_R^{\Delta_k}(U)\left(|s\rangle_B \otimes |x_0\rangle_C\right)=
|s\rangle_B \otimes U^{f_k(s)}|x_0\rangle_C
\]
which maps back under $\Phi^{-1}$ to the same transformed qudit level in the physical same-qudit model. Outside $\mathcal{L}_K$, the physical same-qudit operator acts as the identity by definition.

\paragraph{Abstract Propagation View}

The same correspondence lets one rewrite the propagation stage in the notation of abstract routing labels. In the physical picture, a lifted node on bus $k$ has the form $|w_i + x\cdot\Delta_k\rangle$. Under $\Phi$, this becomes \(|x\cdot\Delta_k\rangle_B \otimes |w_i\rangle_C\).
Hence the source-faithful rule ``if the previous node has bus-$k$ excitation, excite me too'' can also be read abstractly as the propagation of the routing label $|x\cdot\Delta_k\rangle_B$ while the local computational factor $|w_i\rangle_C$ remains intact.

\paragraph{Interpretation}

This tensor-factor extension is useful because it restores the more abstract projector-style algebra without abandoning the physical intuition that motivated the protocol. The same-qudit picture explains where the mechanism comes from, while the tensor-factor picture provides a compact mathematical language for proofs, operator identities, and later generalizations.

\subsection{Invertibility of Routing Operations}

The Controlled-Bus-Load ($CBL$) operator defined earlier \ref{eqn:busload} is not, in general, self-inverse. To ensure reversibility, we explicitly define its inverse.

\paragraph{Inverse Operation}

Define the inverse operator:
\begin{equation}\label{eqn:inversebusload}
    CBL_{(i,j)}^{-\Delta_k} |x\rangle_i |y\rangle_j = |x\rangle_i |(y - x \cdot \Delta_k) \bmod d\rangle_j
\end{equation}
for $x \in \{0,1\}$.

For propagation through an already lifted path, define
\begin{equation}\label{eqn:inversebuspropagate}
    BCP_{(i,j)}^{-\Delta_k} |r\rangle_i |y\rangle_j=|r\rangle_i |(y - f_k(r)\cdot \Delta_k) \bmod d\rangle_j
\end{equation}

Thus, cleanup in the routing protocol of Sec.~\ref{sec:routing-protocol} must use $BCP^{-\Delta_k}$ together with a final inverse initial-lift operation $CBL^{-\Delta_k}$ on the first edge.

\paragraph{Revised Cleanup Step}

The reverse propagation stage becomes:
\[
BCP_{(v_{L-1}, v_L)}^{-\Delta_k}\cdots
BCP_{(v_1, v_2)}^{-\Delta_k}
CBL_{(v_0, v_1)}^{-\Delta_k}
\]

This guarantees exact cancellation of routing information without requiring constraints such as $2\Delta_k \equiv 0 \pmod d$.

\paragraph{Remark}

All routing operations are unitary and reversible, preserving the validity of the quantum circuit model.

\subsection{Compilation Assumption for Routing Primitives}

In the abstract complexity analysis below, we count $CBL$, $BCP$, and $CU_R$ as routing primitives. For a reviewer-facing comparison with qubit routing, it is therefore important to state how these primitives should be interpreted. They are not assumed to be free; rather, they are compiled qudit operations whose cost depends on the available native gate set.

The logical-routing comparison in this paper uses the following compilation assumption: for a fixed bus set and a fixed hardware family with selective multilevel control, each routing primitive can be synthesized using $O(1)$ entangling interaction windows plus local pulses, with device-dependent constants absorbed into the compilation layer. In this viewpoint, $CBL^{\Delta_k}$ is a controlled modular shift, $BCP^{\Delta_k}$ is a controlled shift conditioned on the extracted bus bit, and $CU_R^{\Delta_k}(U)$ is a controlled single-block unitary acting on the lifted target pair. Accordingly, the gate-count comparison in this paper should be read at the level of logical routing primitives. The paper does not prove a hardware-independent constant-depth synthesis theorem for arbitrary qudit platforms; it assumes such a compilation regime when interpreting the logical routing primitives as physical gates.

\begin{center}
\scriptsize
\resizebox{\columnwidth}{!}{%
\begin{tabular}{l l c}
\hline
Primitive & Representative decomposition & Entangling cost \\
\hline
$CBL^{\Delta_k}$ & controlled phase + ladder pulses & $O(1)$ \\
$BCP^{\Delta_k}$ & conditional shift from extracted bus bit & $O(1)$ \\
$CU_R^{\Delta_k}(U)$ & block-selective unitary on the lifted pair & $O(1)$ \\
\hline
\end{tabular}
}
\end{center}

More concretely, a representative compilation of $CBL^{\Delta_k}$ uses one conditioning interaction to distinguish whether the source qudit carries the relevant control state, together with local ladder-type pulses that shift the target into the chosen routing block. Likewise, $BCP^{\Delta_k}$ may be compiled by first conditioning on the relevant extracted bus digit and then applying the corresponding local shift sequence, while $CU_R^{\Delta_k}(U)$ applies a block-selective single-qudit unitary after a routing-bit-selective conditioning stage. These are qualitative synthesis templates rather than hardware-independent exact gate counts, but they make explicit the compilation regime under which each primitive is modeled as a constant number of two-qudit interaction windows plus local pulses.

We therefore do not claim that every device-level synthesis of a routing primitive is uniformly cheaper than three CNOTs in raw pulse count. The intended comparison is architectural: our framework removes explicit SWAP-based state transport and introduces spectral parallelism, while the precise constant factors are delegated to the hardware compilation layer and treated as an assumption of the logical model rather than as a theorem.

The formal model defined here provides the foundation for constructing swap-free routing protocols. 
In subsequent sections, we develop explicit algorithms for non-local gate implementation and analyze their complexity.

\section{Swap-Free Routing Circuits}\label{sec:routing-protocol}

With the routing primitives in hand, we now construct the basic nonlocal circuit. The guiding picture is the same throughout: lift the source bit onto a chosen spectral bus, propagate that bus along a near-neighbor path, act locally on the target block, and then undo the routing excitations so that all intermediate nodes are restored.

\subsection{Problem Definition}

Let $G = (V, E)$ be the connectivity graph. Given two nodes $u, v \in V$ such that $(u,v) \notin E$, our goal is to implement a logical $CX_{(u,v)}$ gate. Let $P = (v_0 = u, v_1, v_2, \dots, v_L = v)$ be a path in $G$ of length $L$ connecting $u$ to $v$.

\paragraph{Scope of Congestion}
It is important to distinguish \emph{spatial routing congestion}, in which independent non-local paths cross at a shared intermediate node, from \emph{algorithmic fan-in}, in which several logical controls are required to influence the same destination within the same logical layer. The proposed parallel-bus architecture resolves the transport aspect of both situations: independent routed signals may traverse shared intermediate nodes and may also arrive simultaneously at a common target qudit without logical mixing, because their bus labels remain algebraically separable. The remaining design choice is then the target semantics. Section~\ref{sec:boolean-fanin} supplies the missing local primitive by defining Boolean-routed fan-in operations on the lifted target block, so the merged framework separates distributed delivery from target-side aggregation.

\subsection{Path Selection and Bus Assignment}

The routing protocol is intended to operate at the compilation level. For a layer of desired nonlocal controlled operations, a transpiler first selects candidate paths in the connectivity graph $G$ and then assigns a routing bus to each path before circuit execution begins.

If two simultaneously scheduled paths intersect, they may not be assigned the same routing bus. Instead, they receive distinct offsets $\Delta_k$, so that their routing information remains distinguishable even while traversing the same physical region of the device. In this sense, congestion is handled by bus assignment rather than by forcing one operation to wait for another to vacate the path.

A simple strategy is to assign to each path the smallest available bus index that is not already used by an intersecting path in the same circuit layer. If no such bus is available among the $K$ supported buses, the corresponding non-local operation is deferred to a later layer. Thus, the routing decisions are fixed during transpilation, while the runtime hardware only executes the prescribed local shift and cleanup operations.

\begin{definition}[Routing Protocol Instance]
A routing protocol instance for $m$ requested non-local controlled operations consists of three pieces of data: a collection of paths
\[
P_j = (v^{(j)}_0, v^{(j)}_1, \dots, v^{(j)}_{L_j}),
\qquad j = 1,\dots,m
\]
where each path connects the source and target of the $j$-th logical operation; a bus-assignment function
\[
b : \{1,\dots,m\} \to \{1,\dots,K\}
\]
such that $b(j) \neq b(\ell)$ whenever $P_j$ and $P_\ell$ intersect; and a schedule
\[
\Sigma = (\mathcal{L}_1,\dots,\mathcal{L}_T),
\]
where each layer $\mathcal{L}_t$ is a set of local routing operations on edges of $G$ and no physical qudit participates in more than one gate in the same layer.
\end{definition}

This definition isolates the three ingredients needed for parallel execution: a chosen path for each routed operation, a collision-free bus labeling on intersections, and a hardware-compatible local schedule.

\paragraph{Illustrative Example}

Suppose the transpiler first assigns a non-local operation from $X_1$ to $X_6$ along the path $X_1 - X_2 - X_3 - X_4 - X_5 - X_6$.
This route is given bus $1$, so its propagation uses the shift $\Delta_1 = 2$.

Now suppose a second non-local operation from $Y_1$ to $Y_2$ intersects the first route near $X_3$ or $X_4$. Because bus $1$ is already occupied on that shared region, the compiler does not reuse it; instead it assigns bus $2$, so the second route uses $\Delta_2 = 4$.

If a third route from $Z_1$ to $Z_2$ also conflicts with the previously assigned paths, it is placed on bus $3$, so it uses $\Delta_3 = 8$.

At a shared node, the resulting qudit value can therefore take the form
\[
x_0 + x^{(1)} 2 + x^{(2)} 4 + x^{(3)} 8,
\qquad x_0, x^{(1)}, x^{(2)}, x^{(3)} \in \{0,1\}
\]
The original control values of the three routes are later recovered independently by the extraction maps $f_1$, $f_2$, and $f_3$. This is the concrete sense in which transpilation removes congestion: the physical qudit may be shared, but the assigned buses remain algebraically distinct.

\subsection{Routing Protocol}

We describe the protocol in three stages: \emph{forward propagation}, \emph{target interaction}, and \emph{cleanup}. The forward stage begins by lifting the source control onto the chosen bus through $CBL_{(v_0, v_1)}^{\Delta_k}$. One then propagates the bus label along the remainder of the path by applying, for $i = 1$ to $L-1$, $BCP_{(v_i, v_{i+1})}^{\Delta_k}$.
The first gate conditionally lifts the first neighbor into bus $k$, and each subsequent $BCP$ gate checks whether the previous qudit carries the bus-$k$ component and, if so, excites the next qudit on the same bus. After this stage, every node on the selected path carries the same bus-$k$ label if and only if the original control bit was $1$.

At the target, apply the routing-controlled operation $CXR^{\Delta_k} = CU_R^{\Delta_k}(X)$ defined in Sec.~\ref{sec:formalism}. For a lifted target qudit with routing offset $s \in \mathcal{S}_K$ and local logical bit $x_0 \in \{0,1\}$, the action is
\[
|s + x_0\rangle \mapsto |s + (x_0 \oplus f_k(s))\rangle
\]
which implements the desired logical CNOT at the target node.

The cleanup stage restores the intermediate nodes by applying the inverse operations in reverse order:
\[
BCP_{(v_{L-1}, v_L)}^{-\Delta_k}
\cdots
BCP_{(v_1, v_2)}^{-\Delta_k}
CBL_{(v_0, v_1)}^{-\Delta_k}
\]
This explicitly reverses the routing transformations and removes all residual information from the intermediate nodes.

\subsection{Worked Example: Non-Local CNOT on a 4-Node Chain}

Consider the path $u - v_1 - v_2 - v$ with source bit $x \in \{0,1\}$ at $u$, target bit $y \in \{0,1\}$ at $v$, and intermediate logical values $w_1,w_2 \in \{0,1\}$ at $v_1$ and $v_2$. Using bus $k$, the protocol begins in the state
\[
|x\rangle_u |w_1\rangle_{v_1} |w_2\rangle_{v_2} |y\rangle_v
\]
After the initial lift, applying $CBL_{(u,v_1)}^{\Delta_k}$ gives
\[
|x\rangle_u |w_1 + x\Delta_k\rangle_{v_1} |w_2\rangle_{v_2} |y\rangle_v
\]
Applying $BCP_{(v_1,v_2)}^{\Delta_k}$ and then $BCP_{(v_2,v)}^{\Delta_k}$ propagates the bus to the rest of the path and yields
\[
|x\rangle_u
|w_1 + x\Delta_k\rangle_{v_1}
|w_2 + x\Delta_k\rangle_{v_2}
|y + x\Delta_k\rangle_v
\]
Applying $CXR^{\Delta_k}$ at the target then gives
\[
|x\rangle_u
|w_1 + x\Delta_k\rangle_{v_1}
|w_2 + x\Delta_k\rangle_{v_2}
|y \oplus x + x\Delta_k\rangle_v
\]
Finally, the inverse propagation and inverse lift return the intermediate nodes to their original states:
\[
|x\rangle_u |w_1\rangle_{v_1} |w_2\rangle_{v_2} |y \oplus x\rangle_v
\]

This explicit 4-node example illustrates the general pattern proved later: the routed bus label propagates forward, triggers the target action, and is then removed without disturbing the logical values stored on the intermediate qudits.

\paragraph{Correctness}

We now prove that the protocol correctly implements a non-local CNOT. The gate $CBL_{(v_0,v_1)}^{\Delta_k}$ first loads bus $k$ onto the path when the source control bit is $1$. Thereafter, each gate $BCP_{(v_i,v_{i+1})}^{\Delta_k}$ checks whether node $v_i$ already contains the bus-$k$ component and propagates the same excitation to node $v_{i+1}$. Consequently, the value at node $v_L$ takes the form \(y + x \cdot \Delta_k\)
where $y$ is the target's original local bit and $x$ is the original control qubit. The operation $CXR^{\Delta_k}$ flips the least-significant logical bit of the lifted target block if and only if bus $k$ is active. Since \(f_k(y + x\cdot\Delta_k) = x\) for $x,y \in \{0,1\}$, this action is equivalent to a standard CNOT. Applying the inverse $BCP$ operations in reverse order, followed by $CBL^{-\Delta_k}$ on the first edge, removes the added routing offsets:
\[
    x \cdot \Delta_k - x \cdot \Delta_k = 0 \mod d
\]
Thus, all nodes except $u$ and $v$ return to their original states.

\paragraph{Complexity Analysis}

Let $L$ be the length of the path between $u$ and $v$.
The forward propagation stage contributes $L$ operations, the reverse cleanup stage contributes another $L$, and the target interaction contributes one more. Thus the total gate count is \(2L + 1\) whereas the transport term of the naive SWAP baseline is $3L$ CNOT-equivalent layers, corresponding to $L$ SWAPs before any additional target-side logic is counted.
Assuming sequential execution along the path, the circuit depth is $2L + O(1)$, and hence $O(L)$ asymptotically.

\subsection{Parallel Routing via Buses}

By assigning distinct buses $\Delta_k$ to independent routing operations, multiple nonlocal gates can be executed simultaneously.

Since routing information is encoded in orthogonal subspaces, operations using different buses do not interfere, even when paths overlap.

This enables parallel routing through local spectral multiplexing rather than through additional ancilla hardware. Under the binary encoding above, the number of simultaneously usable buses grows only logarithmically with $d$, but even modest $K$ can relieve severe spatial contention in compiled routing layers.

\subsection{Congestion-Round Scaling and a Hotspot Separation}

The per-path improvement from $3L$ transport layers to $2L+1$ routed primitives is only part of the story. For a compiled routing layer with chosen paths $P_1,\dots,P_m$, define the \emph{route conflict graph} $\Gamma$ whose vertices are the requested routed operations and whose edges connect precisely those pairs of operations whose chosen paths intersect. Appendix~B proves that, in the fixed-path $K$-bus model, the minimum number of logical routing rounds required for that layer is
\[
R_K(\Gamma)=\left\lceil \frac{\chi(\Gamma)}{K} \right\rceil
\]
where $\chi(\Gamma)$ is the chromatic number of the conflict graph. Thus the comparison extends beyond single-path depth: at the compiler level, congestion management becomes a graph-coloring problem with $K$ reusable local spectral channels rather than a purely spatial serialization problem.

The exact-overlap lower bounds also give a direct necessity statement. Within the exact nearest-neighbor unitary model studied here, any architecture that preserves a resident logical qubit at a congested node and supports $K$ exact overlapping routed controls there must provide local dimension at least $2^{K+1}$. Equivalently, a one-qubit-per-node architecture has zero exact overlap capacity, and a bank of $q$ colocated qubits supports at most $q-1$ such routed bits. Within this exact-overlap model, scalable hotspot relief requires local Hilbert-space expansion: one must add colocated ancillas or enlarge the local Hilbert space.

A simple hotspot family makes the separation explicit. If $m$ requested routes are pairwise intersecting at one shared node, then $\Gamma=K_m$ and the qudit architecture requires exactly
\[
\left\lceil \frac{m}{K} \right\rceil
\]
logical routing rounds. The smallest case is $m=2$. Then $\Gamma=K_2$, two buses suffice for one logical routing round, and already at local dimension $d \geq 8$ a single qudit can preserve the resident logical state while hosting both routed controls exactly. The one-qubit-per-node model fails on that same instance, because it cannot preserve the resident state and simultaneously host even one independently recoverable routed control at the hotspot. Larger clique hotspots simply scale this same separation to $\lceil m/K \rceil$. This example shows that the routing advantage is not limited to a smaller transport constant: it also changes which congested layers can be executed in one round.

\section{Applications and Validation}\label{sec:applications}

This section turns from the basic routed CNOT to two natural extensions. First, the same algebra supports routed Boolean fan-in, where several delivered controls act on a common target through a local same-qudit primitive. Second, the resulting constructions can be checked numerically in both ideal and noisy models, which clarifies the distinction between logical correctness and hardware advantage.

\subsection{Boolean Fan-In Extensions}\label{sec:boolean-fanin}

The routing framework developed so far already solves the \emph{delivery} problem for multiple controls: if several routed signals arrive at a common target on distinct buses, their bus digits remain exactly separable in the lifted target offset. What remains is the \emph{target} problem: which joint logical action should those arrived controls induce on the least-significant target bit? This is no longer a transport question. It is a local same-qudit synthesis problem on the target block.

\subsection{Boolean-Routed Target Primitive}

\begin{definition}[Boolean-Routed Controlled Unitary]\label{def:boolean-routed-cu}
Let
\[
    g : \{0,1\}^{K} \to \{0,1\}
\]
be any Boolean function and let $U$ be any single-qubit unitary. Define the same-qudit Boolean-routed controlled unitary
\[
\begin{aligned}
CU_R^{g}(U)&=
\sum_{s \in \mathcal{S}_K}
\sum_{a,b \in \{0,1\}}\\
&\qquad\times
\left(U^{g(f_1(s),\dots,f_K(s))}\right)_{ab}
\dyad{s+a}{s+b}\\
&\quad+\; I_{\mathcal{R}_K^\perp}
\end{aligned}
\]
\end{definition}

\begin{remark}
The operator $CU_R^g(U)$ acts blockwise on the lifted target qudit. On the block \(\mathrm{span}\{|s\rangle,|s+1\rangle\}\)
it applies either $I$ or $U$ depending on the Boolean value of \(g(f_1(s),\dots,f_K(s))\).
Here $I_{\mathcal{R}_K^\perp}$ denotes the identity on the orthogonal complement of the admissible routing subspace. Thus, outside the admissible routed blocks, $CU_R^g(U)$ acts trivially.
\end{remark}

\begin{proposition}[Action on Lifted Basis States]\label{prop:boolean-target-action}
For every
\[
s = \sum_{k=1}^{K} x_k \cdot\Delta_k \in \mathcal{S}_K
\qquad\text{and}\qquad
x_0 \in \{0,1\},
\]
the Boolean-routed controlled unitary satisfies
\[
CU_R^{g}(U)\ket{s+x_0}=\sum_{b \in \{0,1\}}\left(U^{g(x_1,\dots,x_K)}\right)_{b x_0}
\ket{s+b}
\]
\end{proposition}

\begin{proof}
By exact decoding in the routing algebra,
\[
f_k(s)=x_k
\qquad (k=1,\dots,K)
\]
Substituting this identity into Definition~\ref{def:boolean-routed-cu} gives the stated blockwise action.
\end{proof}

\begin{figure*}[tbp]
\includegraphics[]{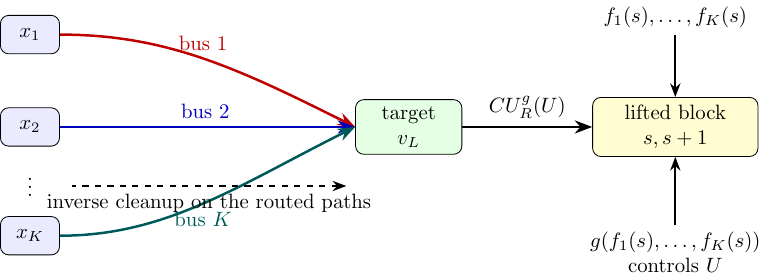}
\caption{
Routed Boolean fan-in in the spectral-routing model. Multiple source controls $x_1,\dots,x_K$ are delivered to a common target along distinct spectral buses, allowing simultaneous routing without spatial duplication. The routed bus pattern is encoded as $s=\sum_k x_k \cdot\Delta_k$ and stored in a lifted two-level block $\{\ket{s},\ket{s+1}\}$. Local digit extraction $f_1(s),\dots,f_K(s)$ feeds a Boolean rule $g$, which controls a same-qudit unitary $CU_R^g(U)$ acting on the logical subspace while preserving the routing offset. The diagram highlights the separation between nonlocal control delivery and local target-side aggregation.
}
\label{fig:routed-fanin-schematic}
\end{figure*}

\begin{corollary}[Single-Bus Recovery]\label{cor:single-bus-recovery}
If $g(x_1,\dots,x_K)=x_k$, then $CU_R^g(U)=CU_R^{\Delta_k}(U)$, the single-bus routed controlled unitary defined earlier.
\end{corollary}

\begin{proof}
The definition reduces immediately to the bus-$k$ control law.
\end{proof}

\begin{corollary}[Toffoli-Type Conjunction]\label{cor:toffoli-type-conjunction}
For any pair of buses $k \neq \ell$, define $g_{k,\ell}(x_1,\dots,x_K)=x_k x_\ell$. Then $CU_R^{g_{k,\ell}}(X)$ flips the target logical bit iff both bus-$k$ and bus-$\ell$ are active.
\end{corollary}

\begin{proof}
Apply Proposition~\ref{prop:boolean-target-action} with $U=X$ and the Boolean rule $g_{k,\ell}$.
\end{proof}

\subsection{Swap-Free Multi-Control Routed Gates}

\begin{theorem}[Swap-Free Multi-Control Routed Unitary]\label{thm:swap-free-multicontrol}
Let \(g : \{0,1\}^{K} \to \{0,1\}\) be a Boolean function and let $U$ be any single-qubit unitary. Suppose the routing stage delivers the control bits $x_1,\dots,x_K$ to the target node, so that immediately before the target interaction the target qudit has the spectral form
\[
    \ket{s+y}\qquad s = \sum_{k=1}^{K} x_k \cdot\Delta_k,\qquad y \in \{0,1\}
\]
Then applying the target primitive $CU_R^g(U)$ followed by the reverse cleanup stage produces the final target state \(U^{g(x_1,\dots,x_K)}\ket{y}\)
while restoring every intermediate routed node to its original logical state.
\end{theorem}

\begin{proof}
By exact decoding \(f_k(s)=x_k\) where \(k=1,\dots,K\).
Hence Proposition~\ref{prop:boolean-target-action} yields
\[
    CU_R^g(U)\ket{s+y}=\sum_{b \in \{0,1\}}\left(U^{g(x_1,\dots,x_K)}\right)_{by}\ket{s+b}
\]
Thus the target logical bit is transformed by the intended Boolean-controlled unitary while the routing offset $s$ is preserved.

The reverse cleanup stage is unchanged from the single-bus case. Each inverse propagation removes exactly the same routing offsets that were written during forward transport, because the relevant bus digits are still present in the routed nodes until the corresponding inverse gate acts. Consequently, all intermediate path nodes return to their initial logical states, and the target qudit returns from the lifted block state $\ket{s+b}$ to the computational state $\ket{b}$. Therefore the final target logical state is \(U^{g(x_1,\dots,x_K)}\ket{y}\) as claimed.
\end{proof}

\begin{corollary}[Routed Boolean Fan-In]\label{cor:routed-boolean-fanin}
Every Boolean fan-in rule
\(g : \{0,1\}^{K} \to \{0,1\}\) admits a swap-free same-qudit routed implementation at the routing level, provided the corresponding source controls can be delivered to the target on distinct buses.
\end{corollary}

\begin{proof}
Choose the target primitive $CU_R^g(U)$ and apply Theorem~\ref{thm:swap-free-multicontrol}.
\end{proof}

\subsection{Complexity Localization and Synthesis}

\begin{definition}[Local Boolean Target Cost]\label{def:local-boolean-target-cost}
For a Boolean function $g : \{0,1\}^{K} \to \{0,1\}$, let $D_g$ denote the depth required to synthesize the target primitive $CU_R^g(U)$ in a chosen local control model on the target qudit. This quantity measures only the local logical fan-in computation at the destination; it does not include transport along the architecture graph.
\end{definition}

\begin{proposition}[Depth Localization]\label{thm:depth-localization}
For fixed fan-in size $K$, suppose a collection of routed controls is delivered to a common target under a schedule whose transport depth is $2L + O(1)$, where $L$ is the maximal routed path length in the layer. Then the corresponding Boolean-routed fan-in gate has total depth $2L + D_g + O(1)$. By contrast, a SWAP-transport implementation using the same target primitive has depth $3L + D_g + O(1)$ under the analogous transport model.
\end{proposition}

\begin{proof}
Theorem~\ref{thm:swap-free-multicontrol} factorizes the routed implementation into three stages: forward transport, a local target action, and reverse cleanup. The forward and reverse transport stages contribute $L$ layers each, up to constant additive scheduling overhead, for a total routing depth of $2L+O(1)$. The target action contributes exactly the local synthesis depth $D_g$. Hence the total routed depth is $2L + D_g + O(1)$. For a SWAP-based implementation, transporting the relevant control information across a path of length $L$ incurs the usual transport depth $3L+O(1)$ before the same local target computation of depth $D_g$ is applied. Therefore the corresponding SWAP baseline has total depth $3L + D_g + O(1)$.
\end{proof}

\begin{remark}
Theorem~\ref{thm:depth-localization} does not claim that Boolean fan-in becomes free. Rather, it localizes the remaining cost: all non-local overhead is carried by the routing term, while all logical aggregation complexity is isolated in the local target term $D_g$.
\end{remark}

\begin{remark}
The hidden constant in the routing term can depend on local scheduling near the target. In particular, if several routed paths terminate at the same target qudit, the final propagation hops may contend on that target and require short serialization. Equivalently, if one wishes to display the fan-in dependence explicitly, one may write the routed depth as $2L + D_g + O(K)$, which reduces to the stated form for fixed $K$.
\end{remark}

\begin{figure*}[tbp]
\centering
\includegraphics[width=0.8\textwidth]{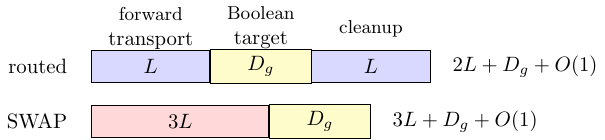}
\caption{
Logical routing cost decomposition for Boolean fan-in. At the level of routing primitives, the spectral-routing framework modifies only the nonlocal transport component: the same local Boolean target cost $D_g$ appears in both implementations, while routed transport contributes $2L+O(1)$ steps compared with the $3L+O(1)$ SWAP-based baseline. The figure highlights the separation between nonlocal routing and local aggregation costs under the constant-cost compilation model assumed throughout.
}
\label{fig:fanin-depth-localization}
\end{figure*}

\begin{definition}[Block-Selective Target Gate]\label{def:block-selective-target-gate}
For each routing offset $s \in \mathcal{S}_K$ and single-qubit unitary $U$, define the block-selective same-qudit operator
\begin{equation}\label{eqn:blockunitary}
    B_s(U)=\sum_{r \notin \{s,s+1\}} \dyad{r}+\sum_{a,b \in \{0,1\}} U_{ab}\dyad{s+a}{s+b}
\end{equation}
\end{definition}

\begin{proposition}[Truth-Table Synthesis Formula]\label{prop:truth-table-synthesis}
Let
\[
    S_g=\left\{s \in \mathcal{S}_K\ \middle|\ g(f_1(s),\dots,f_K(s)) = 1\right\}
\]
Then
\[
    CU_R^g(U)=\prod_{s \in S_g} B_s(U)
\]
The factors commute, and the product contains at most $2^K$ block-selective gates.
\end{proposition}

\begin{proof}
Each operator $B_s(U)$ acts as $U$ on the block $\mathrm{span}\{|s\rangle,|s+1\rangle\}$ and as the identity on every other lifted block. Therefore distinct factors act on disjoint blocks and commute. Their product applies $U$ exactly on those blocks indexed by offsets in $S_g$ and applies the identity elsewhere, which is precisely the definition of $CU_R^g(U)$.
\end{proof}

\begin{corollary}[Generic Upper Bound]\label{cor:generic-upper-bound}
Let $D_{\mathrm{blk}}(U)$ be the depth needed to realize one block-selective gate $B_s(U)$ in a chosen local control model. Then
\[
    D_g \leq |S_g|\,D_{\mathrm{blk}}(U) \leq 2^K D_{\mathrm{blk}}(U)
\]
Hence arbitrary Boolean fan-in always admits an explicit same-qudit synthesis, although in the worst case its local target cost is exponential in $K$.
\end{corollary}

\begin{proof}
Apply Proposition~\ref{prop:truth-table-synthesis} and synthesize each commuting block-selective factor separately.
\end{proof}

\begin{proposition}[Description-Complexity Limitation]\label{prop:description-complexity-limitation}
No uniform synthesis method can realize every Boolean-routed primitive $CU_R^g(U)$ with a target-description size polynomial in $K$ alone. In particular, an unrestricted family of Boolean functions necessarily requires exponential specification in the worst case.
\end{proposition}

\begin{proof}
There are $2^{2^K}$ Boolean functions on $K$ input bits. Any family of synthesis descriptions using only $\mathrm{poly}(K)$ bits can represent at most $2^{\mathrm{poly}(K)}$ distinct target primitives. Since $2^{\mathrm{poly}(K)} \ll 2^{2^K}$ for large $K$, such descriptions cannot cover every Boolean function. Therefore arbitrary $g$ cannot, in general, be specified or synthesized uniformly with only polynomial target description length.
\end{proof}

\begin{remark}
This is why structured Boolean families matter. Spectral routing localizes fan-in to a local primitive, but it does not eliminate the intrinsic complexity of an arbitrary truth table.
\end{remark}

\subsection{Structured Boolean Families}

Let $J \subseteq \{1,\dots,K\}$ be a set of active control buses.

\begin{definition}[Restricted Boolean Families]\label{def:restricted-boolean-families}
For $x=(x_1,\dots,x_K)\in\{0,1\}^K$, define
\[
\wedge_J(x)=\prod_{j\in J} x_j,
\]
\[
\vee_J(x)=
\begin{cases}
1, & \text{if } \sum_{j\in J} x_j \geq 1,\\
0, & \text{otherwise},
\end{cases}
\]
and, for any threshold $t \in \{1,\dots,|J|\}$,
\[
\Theta_{t,J}(x)=
\begin{cases}
1, & \text{if } \sum_{j\in J} x_j \geq t,\\
0, & \text{otherwise}.
\end{cases}
\]
Finally, define the parity rule $\oplus_J(x)=\left(\sum_{j\in J} x_j\right)\bmod 2$.
\end{definition}

\begin{corollary}[Restricted Routed Families]\label{cor:restricted-routed-families}
The target primitives
\[
\begin{aligned}
CU_R^{\wedge_J}(U), \qquad
CU_R^{\vee_J}(U),\\
CU_R^{\Theta_{t,J}}(U), \qquad
CU_R^{\oplus_J}(U)
\end{aligned}
\]
define swap-free routed conjunction, disjunction, threshold, and parity fan-in gates, respectively.
\end{corollary}

\begin{proof}
Each listed Boolean rule is a valid map from $\{0,1\}^K$ to $\{0,1\}$, so the result follows directly from Corollary~\ref{cor:routed-boolean-fanin}.
\end{proof}

\begin{remark}
These families are practically important because conjunction corresponds to Toffoli-type logic, disjunction captures ``any control active'' behavior, threshold rules interpolate between them, and parity captures XOR-type fan-in. In local control models where reversible Boolean logic on the lifted target subspace can be synthesized with polynomial overhead in $|J|$, the corresponding routed implementations have total depth
\[
2L + \mathrm{poly}(|J|),
\]
with the polynomial term absorbing both the local target-synthesis cost and any fan-in-dependent additive scheduling overhead near the target.
\end{remark}

\subsection{Parallel Fan-In Capability and Limits}

\begin{proposition}[What the Construction Provides]\label{prop:what-the-construction-provides}
Once several routed controls have been delivered to the same target qudit in the spectral form $\ket{s+y}$ with $s=\sum_{k=1}^{K} x_k \cdot\Delta_k$, the remaining fan-in computation is entirely local: it depends only on the target-side primitive $CU_R^g(U)$ and is independent of the original source locations.
\end{proposition}

\begin{proof}
After delivery, the source information appears only through the extracted bits $f_k(s)=x_k$. The target action and the reverse cleanup stage therefore depend only on the local lifted target state and the existing routed offsets, not on the geometric positions of the original controls.
\end{proof}

\begin{remark}
This is the conceptual gain of the extension: spectral routing solves distributed delivery, while Boolean-routed target synthesis solves the local aggregation problem at the destination.
\end{remark}

\begin{remark}
What is not proved here is a generic fan-in depth bound beating all locality lower bounds for arbitrary schedules. Delivering several controls to a common target may still incur architecture-dependent serialization near the target neighborhood. The result of this section should therefore be read as a separation between transport and local aggregation, not as a claim of unconditional parallel fan-in speedup on every hardware graph.
\end{remark}

\subsection{Fan-In Sanity Checks}

Small-scale simulations support the algebraic fan-in construction. In the ideal Cirq model, exhaustive truth-table and coherent-superposition tests for a three-control routed conjunction agree with the predicted unitary to machine precision, while ideal-limit QuTiP runs give unit fidelity and zero residual routing population. A first noisy two-control QuTiP stress test shows the expected fidelity degradation once higher-level decay, dephasing, and leakage are introduced. Full numerical details and plots are deferred to Appendix~\ref{app:benchmark-details}.

These fan-in checks also expose an implementation detail that matters for interpretation: in a concrete shared-target schedule, the additive depth constant can exceed the schematic $+1$ because the final propagation hops into a common target cannot all be executed simultaneously. This does not change the linear-in-$L$ scaling, but it indicates that the routing algebra and the final hardware schedule should be distinguished.

\subsection{Physical Mapping Outlook}

The Boolean-routed target primitive turns distributed fan-in into a local multilevel control problem. In a superconducting transmon-like platform, the transported buses correspond to higher energy levels of a weakly anharmonic oscillator, and $CU_R^g(U)$ becomes a block-selective same-qudit unitary whose control condition depends on several extracted spectral digits simultaneously. At an abstract level, such a primitive could be synthesized from number-selective pulses, conditional phases, and short ladder sequences that distinguish the relevant lifted blocks. The main challenge is not formal existence but control complexity: multi-level selectivity and coherence of the higher levels determine whether the target cost $D_g$ remains practical.

In trapped-ion platforms, a high-dimensional internal manifold or encoded local subspace can play the role of the spectral register. The routed buses again arrive as locally readable digits of a single encoded target system, and the Boolean-routed target primitive becomes a local reversible computation followed by a conditioned logical rotation on the target block. Here the advantage is stronger local controllability; the open question is how best to compile the resulting multi-level conditional action into experimentally natural pulse sequences.

\begin{remark}
From a hardware-design viewpoint, the algebraic fan-in result should be read as a division of tasks. Spectral routing solves distributed delivery. Physical implementation then reduces to the synthesis of a local multi-level Boolean-controlled target gate.
\end{remark}

\subsection{Numerical Benchmarking and Algorithmic Scaling}\label{sec:numerical-benchmarking}

To validate the ideal guarantees of the routing protocol on small instances, we simulated the architecture in Google Cirq using exact state-vector evolution. Exact validation is reported only for tractable instances satisfying the dimension and state-vector requirements; larger sweeps are interpreted as structural scaling evidence rather than as full-state certification.

The validated ideal runs support the main claims: the routed nonlocal CNOT matches the target unitary to machine precision, overlapping routes show zero crosstalk on the tested instance, and compiled depth follows the expected $2L+1$ pattern while the naive SWAP transport baseline follows $3L$. Full benchmark data and plots are collected in Appendix~\ref{app:benchmark-details}.

\subsection{Compiler-Level Route-Demand Benchmarks}

To test whether the congestion picture persists beyond toy paths, we ran a separate compiler-level route-demand benchmark on four nontrivial circuit families: QFT, regular-graph QAOA, a Grover-style amplitude-amplification circuit, and a mirror-interaction stress test. This choice deliberately overlaps the application set emphasized in recent swapless parity-tracking work~\cite{swapless2024}, namely QFT and QAOA, while adding one additional workload and one congestion-stress family. The methodology was intentionally bias-controlled. For each logical two-qubit layer, routes were chosen only from topology shortest paths, shortest-path ties were sampled uniformly at a fixed seed, and qubit-transpiler reference statistics were reported separately from the routed logical metrics. The amplitude-amplification instance uses one Grover-style iteration with a fixed marked basis state and is decomposed to ordinary one- and two-qubit gates before route extraction. Thus the benchmark measures exactly the claims made here---transport demand, conflict-graph coloring, and routing rounds---without hand-tuning paths in favor of the qudit model.

This also clarifies the scope of the comparison with the recent swapless QFT and QAOA results. Their analysis reports exact qubit-circuit resource counts after an application-specific compilation, whereas our present data target route demand, conflict-graph coloring, and routed transport under fixed shortest-path routing. The two comparisons are therefore not identical gate-for-gate. What can be compared directly is the shared application layer together with the route-demand metric imposed on that layer. On the shared QFT and QAOA families, our data show a smaller routed transport term under this fixed routing metric, while the spectral-bus model provides an additional congestion-round description when several routed interactions overlap.

The linear-chain rows are the direct LNN overlap with that application study. At $n=8$, the QFT line case has mean SWAP transport $19.38$ versus mean routed transport $15.08$, for a ratio of $0.778$, while the QAOA line case has mean SWAP transport $15.69$ versus mean routed transport $12.31$, for a ratio of $0.784$. In the same LNN setting, the measured mean conflict demand is $\chi(\Gamma)=2.15$ for QFT and $\chi(\Gamma)=1.69$ for QAOA, so two spectral buses reduce the mean routing rounds to $1.38$ and $1.00$, respectively. Thus, on the same QFT/QAOA application layer emphasized by the recent SWAP-less paper, the present benchmark yields a smaller routed transport term under the fixed shortest-path-demand metric used here and additionally exposes the congestion-round structure.

The resulting workloads realize the predicted compiler-level separation. On the representative $n=8$ cases in Table~\ref{tab:compiler-route-benchmarks}, the routed logical transport term stays between $0.75$ and $0.89$ of the corresponding SWAP transport baseline, while the measured route-conflict graphs produce mean chromatic numbers between $1.04$ and $4.0$. For $K=2$ buses, the corresponding mean congestion rounds range from $1.0$ to $2.0$. In the most severe mirror-line case, the mean conflict demand is $\chi(\Gamma)=4$ and two spectral buses reduce the mean routing rounds from $4$ to $2$. By contrast, the amplitude-amplification workload produces only mild overlap, with routed-to-SWAP transport ratios of $0.816$ on the line and $0.863$ on the grid at $n=8$, and mean two-bus congestion rounds equal to $1.0$ on both topologies. In parallel, an independent exact validator found zero discrepancies with the formula $R_K(\Gamma)=\lceil \chi(\Gamma)/K \rceil$ on all $1099$ labeled conflict graphs up to five vertices. Detailed benchmark records are collected in Appendix~\ref{app:benchmark-details}.

\begin{table*}[tbp]
\centering
\scriptsize
\begin{tabular}{l l c c c c c}
\hline
Family & Topology & mean SWAP transport & mean routed transport & ratio & mean $\chi(\Gamma)$ & mean $R_2$ \\
\hline
Amplitude ($n=8$) & line (LNN) & 7.42 & 6.06 & 0.816 & 1.04 & 1.00 \\
Amplitude ($n=8$) & grid & 5.67 & 4.89 & 0.863 & 1.04 & 1.00 \\
QFT ($n=8$) & line (LNN) & 19.38 & 15.08 & 0.778 & 2.15 & 1.38 \\
QFT ($n=8$) & grid & 12.46 & 10.46 & 0.840 & 1.85 & 1.23 \\
QAOA ($n=8$) & line (LNN) & 15.69 & 12.31 & 0.784 & 1.69 & 1.00 \\
QAOA ($n=8$) & grid & 8.31 & 7.38 & 0.889 & 1.15 & 1.00 \\
Mirror ($n=8$) & line (LNN) & 48.00 & 36.00 & 0.750 & 4.00 & 2.00 \\
Mirror ($n=8$) & grid & 24.00 & 20.00 & 0.833 & 3.00 & 2.00 \\
\hline
\end{tabular}
\caption{Bias-controlled compiler-level route-demand benchmarks on representative nontrivial circuit families. The rows labeled line (LNN) are the direct geometry overlap with recent swapless QFT/QAOA application benchmarks, while amplitude amplification and the mirror family test additional regimes. For each logical two-qubit layer, routes are assigned using only topology shortest paths, with shortest-path ties sampled uniformly at a fixed seed. The routed transport term is consistently below the corresponding SWAP transport baseline on the same path demand, while the measured conflict graphs realize the predicted two-bus congestion scaling.}
\label{tab:compiler-route-benchmarks}
\end{table*}

\section{Physical Realization, Gate Sets, and Limitations}\label{sec:physical-realization}

\subsection{Qudit Implementations}

Many quantum computing platforms naturally support multi-level systems. In particular, superconducting transmon qubits exhibit multiple accessible energy levels beyond the computational basis \cite{koch2007charge}. Experimental works have demonstrated control over 3–7 energy levels, making them suitable candidates for implementing the routing subspace $\mathcal{H}_R$ \cite{Goss2024, chi2022programmable}. At the level of control primitives, the abstract operators $CBL$ and $BCP$ may be viewed as compiled unitaries rather than literally elementary gates. In a transmon-like platform, a representative realization is to synthesize the required conditional level shift using number-selective or state-selective pulses between neighboring levels, interleaved with controlled phase rotations generated by the native coupling between adjacent devices \cite{heeres2017implementing, smith2022programming}. In this picture, ladder-type selective pulses move amplitude through the relevant level pairs, while conditional phases ensure that the shift is applied only when the control qudit occupies the required computational or bus-resolved state. The precise pulse decomposition is hardware dependent, but the required ingredients are standard qudit-control capabilities: selective addressability of higher transitions together with conditional interaction phases.In superconducting transmons, however, the weak anharmonicity of the energy ladder makes higher transitions increasingly crowded and harder to address selectively as $d$ grows \cite{koch2007charge}. This hardware reality directly constrains the number of reliable routing buses that can be supported in practice. Trapped ions \cite{ringbauer2022universal, shi2025efficient} and photonic systems \cite{wang2018multidimensional} also provide native qudit capabilities, further supporting the feasibility of the proposed model.

\subsection{Advantages of Qudit Routing}

The proposed approach offers three immediate advantages. It removes explicit SWAP transport from the routing primitive and therefore reduces the leading transport overhead. It also permits several routing processes to coexist through distinct routing levels, so the architecture addresses congestion as well as path length. Finally, it does not require additional physical qudits at the routed sites, because the extra routing capacity is drawn from the higher levels of the same local device rather than from separate ancillas.

\paragraph{Non-Destructive Routing on Dense Grids}
A SWAP gate physically exchanges the quantum states of two adjacent nodes. In dense circuit layers, many intermediate qudits already carry computational data whose placement matters for subsequent local gates. Routing a control signal through such a node by SWAP therefore displaces the resident data and generally requires compensating transport, often a reverse SWAP cascade, to restore the original layout before later layers can proceed. The proposed protocol avoids this displacement entirely. By using the higher-energy routing subspace, the forward propagation and reverse cleanup phases act non-destructively on the intermediate computational states while preserving their original locations. The resulting $O(L)$ cleanup cost is therefore the temporal price paid to preserve the spatial data locality of the surrounding algorithm.

\subsection{Noisy QuTiP Benchmarking}

To probe the physical limits of the routing idea beyond the ideal Cirq checks of Sec.~\ref{sec:numerical-benchmarking}, we simulated open-system dynamics in QuTiP using Lindblad evolution and, for the larger archived runs, Monte Carlo trajectories. The detailed datasets are reported in Appendix~\ref{app:benchmark-details}.

The main conclusion is as follows. Under the present higher-level lifetime assumptions, the routed protocol does not outperform the SWAP baseline in fidelity, even though it improves the logical transport term from $3L$ to $2L+1$. At the same time, threshold scans show that this conclusion is conditional rather than fundamental: improving the higher-level coherence times and/or reducing the routed primitive duration eventually opens a regime in which the routed architecture becomes favorable.

\subsection{Limitations}

Despite these advantages, the approach also introduces several challenges. Leakage is an immediate concern: operations involving higher energy levels may populate states outside the intended computational and routing subspaces, and such errors are typically harder to correct than standard qubit errors. Gate fidelity is a second concern, because control over higher energy levels is usually less precise than operations restricted to $\mathcal{H}_C$, so overall performance may degrade unless those transitions are carefully calibrated.
Hardware constraints provide the third limitation.
The number of available routing buses is limited by the qudit dimension $d$. Because the binary encoding requires $d \geq 2^{K+1}$ to support $K$ parallel buses while preserving the logical bit in the least-significant position, the supported bus count scales only logarithmically with dimension: $K \leq \lfloor \log_2 d \rfloor - 1$.
Thus the available routing parallelism is not linear in $d$ but only $O(\log d)$. In particular, a device with only $d=5$ or $d=7$ reliably accessible levels supports at most $K=1$ routing bus. In transmon hardware, weak anharmonicity and reduced pulse selectivity make larger effective values of $d$ progressively harder to exploit without leakage.

Finally, crosstalk between nearby transitions may prevent perfect isolation of the addressed levels, leading to unintended interactions during routing operations.

\subsection{Discussion}

The proposed framework is best viewed as an architectural abstraction that trades increased control complexity for reduced routing overhead. Conceptually, it introduces a routing layer that converts spatial congestion into spectral separation, treats each qudit as a multiplexed routing node that stores local logical data while carrying routed signals in higher levels, enables concurrent non-local routing through orthogonal bus labels rather than additional ancilla hardware, and separates the source-faithful same-qudit physical picture from a cleaner tensor-factor algebraic formalism.

The protocol should therefore be read as a forward-looking architectural proposal rather than as a claim of immediate advantage on current devices. In this framing, the noisy QuTiP study is useful precisely because it maps the regime of hardware parameters where the ideal depth and congestion advantages begin to translate into operational fidelity gains.

\paragraph{Practical Parallelism Ceiling}
The algebraic limit $K \leq \lfloor \log_2 d \rfloor - 1$ is a strict ceiling, but it should not be read as a requirement to pursue arbitrarily large local dimension before the routing idea becomes useful. In a given circuit layer, the number of routing buses actually required is determined by the conflict graph induced by overlapping routing paths: a bus assignment exists whenever that conflict graph can be properly colored with at most $K$ colors, so no constant-$K$ guarantee holds in the worst case. In locality-constrained architectures, however, compiled routing layers may still exhibit modest overlap structure, so even a small number of routing buses can relieve conflicts that would otherwise force serialization or deferred scheduling. We therefore view the logarithmic scaling as a practical congestion-relief mechanism rather than a route to unbounded parallelism: it can provide a small constant-factor increase in effective routing concurrency without requiring additional physical qudits. Conceptually, this is analogous to introducing a small number of virtual channels in classical network routing.

\paragraph{Materials Outlook}
The present QuTiP data makes the next experimental target clear: larger routing capacity will matter only if the higher levels remain coherent enough to survive the forward-propagate, target, and cleanup cycle. Realizing more ambitious capacities such as $d=16$ or $d=32$ therefore calls for platforms with flatter effective anharmonicity, selective access to many neighboring transitions, and long higher-level $T_\phi$ times. Single-Molecule Magnets are especially interesting candidates in this respect, because their large internal spin manifolds naturally provide high local dimension without introducing separate ancilla hardware.

Future work may therefore address error-mitigation strategies for routing subspaces, improved encoding schemes for bus allocation, and hybrid compilation strategies that combine SWAP transport with qudit routing when the hardware regime favors a mixed approach.

\section{Conclusion and Outlook}

We have presented a spectral-routing framework for near-neighbor quantum hardware in which higher qudit levels act as orthogonal buses for nonlocal control delivery. In the ideal model, the construction implements a routed nonlocal controlled operation in $2L+1$ logical routing primitives over a path of length $L$, preserves the resident data on intermediate nodes, and allows overlapping routes to remain exactly separable through bus labels. The same picture extends naturally to Boolean fan-in: once several controls are delivered to a common target, the remaining task is a local same-qudit synthesis problem rather than an additional nonlocal routing problem.

The main separation is therefore not only the reduction from $3L$ transport layers to $2L+1$ routed primitives. At the compiler level, parallelism is limited by graph coloring, and spectral routing increases the number of locally available colors by supplying $K$ reusable buses at one physical site. Quantitatively, the routing-round complexity becomes
\[
R_K(\Gamma)=\left\lceil \frac{\chi(\Gamma)}{K} \right\rceil,
\]
and the smallest $m=2$ hotspot already shows that exact same-node overlap routing can fail in the one-qubit-per-node model while remaining one-round in the qudit model once $d \geq 8$. Within the exact nearest-neighbor unitary exact-overlap model studied here, that is the sense in which local Hilbert-space expansion is necessary for exact congestion relief.

The numerical results clarify the boundary between logic and hardware. Exact Cirq simulations confirm that the routed constructions are correct on the validated instances, including complete cleanup and zero crosstalk in the overlapping-path regime. The QuTiP study gives the complementary message: reduced transport depth does not automatically translate into higher fidelity. Under the present higher-level lifetime assumptions, the routed protocol remains fidelity-inferior to the SWAP baseline, but the threshold scans identify the coherence-speed regime in which the reduced transport depth is expected to become operationally useful.

Separate exact graph checks and compiler-level route-demand benchmarks support the same conclusion on the logical side. The graph validator finds zero discrepancies with the congestion theorem on all tested small conflict graphs, while the QFT, QAOA, amplitude-amplification, and mirror workloads realize the predicted transport and coloring advantages without any path selection beyond shortest-path routing.

Several directions follow immediately from this work. On the theory side, one would like sharper synthesis results for structured Boolean target families and better routing heuristics for conflict-graph-aware bus assignment. On the hardware side, the main challenge is clear from the noisy simulations: spectral routing becomes advantageous only when higher-level control, leakage suppression, and higher-level coherence improve enough to support the forward-propagate, target, and cleanup cycle. In that sense, the proposal is best read as a routing architecture for emerging qudit hardware rather than as a claim about immediate superiority on present-day devices.

\bibliography{main}

\appendix

\section{Formal Guarantees for Routed Control Delivery}

This appendix collects the formal correctness and resource statements for the routed single-control protocol and the basic parallel-routing abstraction.

\subsection{Preliminaries}

Let $G = (V, E)$ be a connected graph representing the hardware architecture, and let $u, v \in V$ be two nodes connected by a path $P = (v_0 = u, v_1, \dots, v_L = v)$ of length $L$. Let $K$ denote the number of simultaneously addressable routing buses, with offsets $\Delta_k = 2^k$ for $k=1,\dots,K$. The largest routing value produced by activating every bus while retaining the computational bit in the least-significant position is $1 + \sum_{k=1}^{K} \Delta_k = 2^{K+1} - 1$. As stated earlier in the no-aliasing assumption, we therefore require $d \geq 2^{K+1}$, or equivalently $d > 1 + \sum_{k=1}^{K} \Delta_k$, so that no modular wrap-around occurs during routing.

\subsection{Correctness of the Routed Nonlocal CNOT}

\begin{theorem}[Swap-Free Non-Local CNOT]\label{thm:swap-free-nonlocal-cnot}
Let $u, v \in V$ be two nodes connected by a path of length $L$. 
Using the routing protocol defined in Sec.~\ref{sec:routing-protocol}, a logical $CX_{(u,v)}$ operation can be implemented without SWAP gates, while restoring all intermediate nodes to their original states.
\end{theorem}

\begin{proof}
Let the initial computational values along the path be $x \in \{0,1\}$ at the source node $v_0 = u$, values $w_i \in \{0,1\}$ at the intermediate nodes $v_i$ for $i=1,\dots,L-1$, and $y \in \{0,1\}$ at the target node $v_L = v$. The initial state is therefore
\[
|\Psi_0\rangle=|x\rangle_{v_0}\bigotimes_{i=1}^{L-1} |w_i\rangle_{v_i}|y\rangle_{v_L}
\]
and every qudit lies in the computational block $\{|0\rangle, |1\rangle\}$ of the same physical qudit architecture.

\paragraph{Forward propagation}

The first gate in the forward stage is the initial lift $CBL_{(v_0,v_1)}^{\Delta_k}$, which sends node $v_1$ to $|w_1 + x\Delta_k\rangle$.

Now consider a subsequent edge $(v_t,v_{t+1})$ with $t=1,\dots,L-1$. If node $v_t$ is in the state $|w_t + x\cdot\Delta_k\rangle$, then because $w_t \in \{0,1\}$ we have $f_k(w_t + x\cdot\Delta_k) = x$.
Hence the propagation gate $BCP_{(v_t,v_{t+1})}^{\Delta_k}$ adds the same shift $x\cdot\Delta_k$ to node $v_{t+1}$.

By induction along the path, after the full forward stage we obtain
\[
|\Psi_{\mathrm{fwd}}\rangle=|x\rangle_{v_0}\bigotimes_{i=1}^{L-1} |w_i + x \cdot \Delta_k\rangle_{v_i}|y + x \cdot \Delta_k\rangle_{v_L}
\]
Thus every node on the selected route preserves its original least-significant logical bit while acquiring the same bus-$k$ label in the higher levels.

\paragraph{Target interaction}

At the target node, the lifted state belongs to the two-level block indexed by the routing offset $s = x\cdot\Delta_k$. Since $f_k(s) = x$, the routing-controlled NOT acts as $CXR^{\Delta_k}|y + x \cdot \Delta_k\rangle_{v_L} = |y \oplus x + x \cdot \Delta_k\rangle_{v_L}$, which flips the target logical bit if and only if the original control bit was $1$.

\paragraph{Cleanup}

Finally, apply the inverse cleanup sequence
\[
    BCP_{(v_{L-1},v_L)}^{-\Delta_k}\cdots BCP_{(v_1,v_2)}^{-\Delta_k} CBL_{(v_0,v_1)}^{-\Delta_k}.
\]
At each step, the source qudit still contains the same bus-$k$ component $x$, so the inverse gate subtracts $x\Delta_k$ from the next node. Therefore all intermediate nodes return from $|w_i + x\cdot\Delta_k\rangle$ to $|w_i\rangle$, and the target returns from $|y \oplus x + x\cdot\Delta_k\rangle$ to $|y \oplus x\rangle$.

\paragraph{Conclusion}

\[
|\Psi_f\rangle=|x\rangle_{v_0}
\bigotimes_{i=1}^{L-1} |w_i\rangle_{v_i}|y \oplus x\rangle_{v_L}
\]
which is exactly the action of a nonlocal $CX_{(u,v)}$ gate on the same-qudit higher-level routing model.
\end{proof}

\subsection{Complexity Bounds}

\begin{proposition}[Gate Complexity]\label{prop:gate-complexity}
The routing protocol implements a nonlocal CNOT over a path of length $L$ using $2L + 1$ primitive operations.
\end{proposition}

\begin{proof}
The protocol consists of $L$ forward propagation steps, one target interaction, and $L$ reverse cleanup steps, yielding a total of $2L + 1$ operations.
\end{proof}

\begin{proposition}[Depth]\label{prop:protocol-depth}
The circuit depth of the protocol is $2L + O(1)$, and hence $O(L)$, under sequential execution.
\end{proposition}

\begin{proof}
The forward stage contributes $L$ sequential routing steps, the target interaction contributes one constant-depth step, and the cleanup stage contributes another $L$ sequential routing steps. Thus the total depth is $2L + O(1)$.
\end{proof}

\subsection{Parallel Routing Capacity}

\begin{theorem}[Bus Parallelism]\label{thm:bus-parallelism}
Let the available routing buses be $\Delta_k = 2^k$ for $k = 1,\dots,K$, and suppose that $d \geq 2^{K+1}$. Then up to $K$ independent nonlocal routing operations can share nodes while remaining algebraically separable by their bus labels. In particular, overlapping routes assigned distinct buses do not suffer logical mixing at the level of the ideal routing-value model.
\end{theorem}

\begin{proof}
At any shared node, simultaneous routing operations produce the aggregate routing value
\[
    r = x_0 + \sum_{k=1}^{K} x_k 2^k, \qquad x_0, x_k \in \{0,1\}
\]
Because $0 \leq r \leq 2^{K+1} - 1 < d$, no modular wrap-around occurs. Moreover, the $k$-th routed control bit is recovered exactly by
\[
    f_k(r) = \left\lfloor \frac{r}{2^k} \right\rfloor \pmod 2 = x_k
\]
Hence the contribution of each bus remains separable inside the combined routing value, and overlapping paths do not corrupt one another at the level of bus-label algebra.
\end{proof}

This independence statement holds at the level of the ideal unitary model and establishes algebraic multiplexing rather than a full hardware-scheduling claim. In physical implementations, leakage, crosstalk, and other cross-level errors can spoil exact bus independence; these limitations are discussed in Sec.~\ref{sec:physical-realization}. The corresponding scheduled-execution statement is addressed separately in Theorem~\ref{thm:parallel-correctness}.

\begin{theorem}[Parallel Correctness]\label{thm:parallel-correctness}
Let $\mathcal{I} = \left(\{P_j\}_{j=1}^{m}, b, \Sigma \right)$ be a routing protocol instance. Assume that each path $P_j$ is executed using the routing protocol of Sec.~\ref{sec:routing-protocol} on bus $b(j)$, and that intersecting paths receive distinct bus assignments. Then the scheduled global execution implements the intended collection of nonlocal controlled operations correctly: each routed target receives the same logical action it would receive in the corresponding single-path protocol, and all routing excitations are removed at the end of the schedule.
\end{theorem}

\begin{proof}
Fix a shared node $v$. Because intersecting paths are assigned distinct buses, the routing value at $v$ always has the form
\[
    r_v = x_{0,v} + \sum_{j : v \in P_j} x_j \cdot\Delta_{b(j)}
\]
where $x_{0,v}$ is the local logical bit stored at $v$ and each $x_j \in \{0,1\}$ is the control value associated with the $j$-th routed operation. By the binary encoding, the contribution of route $j$ is recovered exactly by
\(f_{b(j)}(r_v) = x_j\).
Therefore, even when several routed operations overlap at $v$, each route continues to read only its own bus digit.

Now consider the scheduled execution. Since each layer in $\Sigma$ acts on disjoint local supports, the corresponding layer unitary is well defined. For a fixed route $j$, the gates belonging to that route perform exactly the same three tasks as in the single-path protocol: the forward stage writes the offset $\Delta_{b(j)}$ onto the selected path when $x_j=1$, the target gate applies $X^{x_j}$ to the logical target bit, and the cleanup stage removes the same offset. Because no other route uses bus $b(j)$ on any shared node, operations from other routes may add or remove only different binary digits of the same qudit value; they do not alter the digit extracted by $f_{b(j)}$.

Hence the single-path correctness proof applies route-by-route inside the global execution. Each requested routed operation acts correctly on its designated source and target, and after all cleanup stages every auxiliary bus contribution is cancelled. Thus the final state contains exactly the intended logical nonlocal controlled operations together with no residual routing information.
\end{proof}

\subsection{Comparison with SWAP-Based and Alternative Routing}

\begin{proposition}[Reduction in Leading Routing Constant]\label{prop:swap-routing-improvement}
For a path of length $L$, the proposed method realizes the routed operation using $2L+1$ logical routing primitives rather than the transport-only SWAP baseline of depth $3L$, while requiring no additional physical qudits.
\end{proposition}

\begin{proof}
Naive SWAP routing requires $L$ swaps to move a qubit across a path, and each swap costs 3 CNOT gates, giving a transport term of $3L$ before any additional target-side logic is counted. In contrast, the proposed method directly propagates information using $2L + 1$ logical routing primitives without physically relocating the resident data.
\end{proof}

\begin{table}[ht]
    \centering
    \scriptsize
    \resizebox{\columnwidth}{!}{%
    \begin{tabular}{l c c c}
    \hline
    Method & Extra qudits & Depth & Parallelism \\
    \hline
    Naive SWAP routing & none & $O(L)$ & limited \\
    Ancilla routing & additional ancillas & $O(L)$ & moderate \\
    This work & none & $2L + O(1)$ & up to $K$ buses \\
    \hline
    \end{tabular}
    }
    \caption{Routing cost}
    \label{tab:placeholder}
\end{table}

\begin{center}

\end{center}

This comparison highlights the main architectural tradeoff of the present proposal. Standard qubit routing keeps the local dimension fixed but suffers from spatial contention, ancilla-based schemes relieve some of that contention by adding hardware, and the present framework instead reuses the same physical qudits while opening concurrent routing channels in their higher-dimensional subspaces.

Teleportation-based and other entanglement-assisted routing strategies replace spatial movement with pre-shared entanglement, Bell-type measurements, classical feedforward, and the associated resource preparation overhead. Likewise, measurement-based approaches shift the routing burden into the preparation of a suitable entangled resource state together with adaptive measurement patterns. By contrast, the present method remains within unitary nearest-neighbor evolution on the underlying hardware graph and does not require pre-distributed entanglement or measurement-conditioned correction. The tradeoff is that our advantage depends on the coherence and controllability of higher qudit levels rather than on an external entanglement supply.

\section{Lower Bounds and Congestion Optimality}

This appendix collects lower bounds and exact congestion-round statements for the fixed-path $K$-bus routing model.

\subsection{Lower Bound on Information Propagation}

\begin{theorem}[Light-Cone Lower Bound]\label{thm:lightcone}
Let $G = (V, E)$ be a graph representing a quantum architecture with nearest-neighbor interactions. 
Consider two nodes $u, v \in V$ separated by graph distance $L$.
Any protocol that implements a non-local $CX_{(u,v)}$ gate using only local interactions requires depth at least $L$.
\end{theorem}

\begin{proof}
We argue using a causality (light-cone) argument using the Lieb-Robinson bound\cite{lieb1972finite} for discrete quantum circuits. Let $A_u$ be a local operator supported at node $u$ (e.g., the control bit's $Z$ operator) and $B_v$ be a local operator at node $v$. In a system with nearest-neighbor interactions, the Heisenberg-picture evolution of an operator $O$ after $t$ time steps is given by $O(t) = \mathcal{U}_t^\dagger O \mathcal{U}_t$, where $\mathcal{U}_t$ is the unitary circuit of depth $t$.

The Lieb-Robinson bound states that the support of $A_u(t)$ grows at a finite velocity. Specifically, for a circuit of depth $t$ composed of gates acting on edges in $E$, the support of $A_u(t)$ is strictly contained within the ball $B(u, t) = \{w \in V \mid dist(u, w) \leq t\}$. 

A non-local $CX_{(u,v)}$ gate requires that the state at node $v$ becomes conditionally dependent on the state at node $u$. This implies that there exist local operators $A_u$ and $B_v$ such that their evolved commutator is non-vanishing:
\[ [A_u(t), B_v] \neq 0 \]
However, the Lieb-Robinson bound implies that if $dist(u, v) > t$, then the support of $A_u(t)$ and $B_v$ are disjoint, and thus $[A_u(t), B_v] = 0$. For a non-trivial interaction to occur at distance $L$, we must therefore have $t \geq L$. 

This bound is a fundamental property of local quantum evolution and holds regardless of the local Hilbert space dimension $d$, the use of ancillas, or the specific gate decomposition (e.g., SWAP vs. spectral routing). Thus, any implementation of a non-local CNOT over distance $L$ requires depth $t \geq L$.

\end{proof}

\subsection{Optimality of the Proposed Protocol}

\begin{theorem}[Asymptotic Optimality]
The routing protocol defined in Sec.~\ref{sec:routing-protocol} achieves depth $O(L)$ for implementing a non-local $CX_{(u,v)}$ gate. Therefore, it is asymptotically optimal with respect to circuit depth.
\end{theorem}

\begin{proof}
From Sec.~\ref{sec:routing-protocol}, the protocol consists of $L$ forward propagation steps and $L$ reverse cleanup steps. These operations can be arranged sequentially with total depth $2L + O(1)$. Since Theorem~\ref{thm:lightcone} establishes a lower bound of $L$, the protocol matches the lower bound up to constant factors.
\end{proof}

\subsection{Parallelism vs. the Lower Bound}

While the depth lower bound applies to a single nonlocal operation, the proposed protocol allows multiple operations to proceed simultaneously.

\begin{proposition}[Parallel Multiplexing Advantage]
Using $k$ distinct routing buses, the protocol enables up to $k$ independent nonlocal operations to be multiplexed within a common logical routing layer without logical bus interference, provided a compatible local schedule exists.
\end{proposition}

\begin{proof}
Each routing operation uses a distinct offset $\Delta_i$, ensuring orthogonality in the routing subspace. Since these encodings do not interfere at the level of the routing-value algebra, multiple operations may share nodes while preserving separable bus labels. Theorem~\ref{thm:parallel-correctness} then shows that any schedule respecting the local interaction constraints implements the same logical action as the corresponding collection of single-path protocols. Thus bus multiplexing does not by itself introduce additional per-route routing depth beyond the schedule already imposed by local edge constraints.
\end{proof}

\subsection{Conflict-Graph Separation from Qubit-Only Routing}

\begin{definition}[Route Conflict Graph]
Given a requested family of routed operations with paths $P_1,\dots,P_m$, define the \emph{route conflict graph} $\Gamma$ to be the graph with vertex set $\{1,\dots,m\}$ in which $j$ and $\ell$ are adjacent iff $P_j \cap P_\ell \neq \varnothing$. Thus an edge of $\Gamma$ records exactly the pairs of routed operations that cannot share the same bus in one logical routing round.

Its chromatic number is denoted $\chi(\Gamma)=\min\{c : \Gamma \text{ admits a proper } c\text{-coloring}\}$.
\end{definition}

\begin{proposition}[Single-Round Bus-Assignment Criterion]\label{prop:single-round-bus-assignment}
Fix $K$ available routing buses. A requested family of routes can be placed in one logical routing round of the same-qudit architecture if and only if its route conflict graph $\Gamma$ is $K$-colorable, provided the chosen color classes admit a compatible local edge schedule.
\end{proposition}

\begin{proof}
A proper $K$-coloring $c : \{1,\dots,m\} \to \{1,\dots,K\}$ is exactly a bus-assignment function with $c(j) \neq c(\ell)$ whenever $P_j \cap P_\ell \neq \varnothing$.
This is precisely the collision-free bus-assignment condition in the routing model. Conversely, any feasible one-round execution assigns one bus to each route and must assign distinct buses to intersecting routes, so the used bus labels define a proper $K$-coloring of $\Gamma$.
\end{proof}

\begin{theorem}[Optimal Congestion-Round Complexity]\label{thm:congestion-round-complexity}
Let $\Gamma$ be the route conflict graph of a requested routing layer, and let $\chi(\Gamma)$ denote its chromatic number. Then, at the bus-assignment level, the minimum number of logical routing rounds needed in the same-qudit $K$-bus model is
\[
R_K(\Gamma)
=
\left\lceil \frac{\chi(\Gamma)}{K} \right\rceil,
\]
provided each round admits a compatible local edge schedule.
\end{theorem}

\begin{proof}
Suppose the routes are executed in $R$ logical rounds. By Proposition~\ref{prop:single-round-bus-assignment}, the routes assigned to any one round induce a $K$-colorable subgraph of $\Gamma$. Coloring each round with its own palette of $K$ colors therefore produces a proper coloring of all of $\Gamma$ with at most $RK$ colors. Hence
\[
\chi(\Gamma) \leq RK,
\]
so every implementation must satisfy $R \geq \left\lceil \frac{\chi(\Gamma)}{K} \right\rceil$.

For the matching upper bound, choose a proper $\chi(\Gamma)$-coloring of $\Gamma$ and partition its color classes into consecutive groups of size at most $K$. Each group defines one routing round whose induced conflict graph is $K$-colorable, so Proposition~\ref{prop:single-round-bus-assignment} realizes that group in a single logical round. Therefore
\[
R_K(\Gamma)
\leq
\left\lceil \frac{\chi(\Gamma)}{K} \right\rceil,
\]
which matches the lower bound.
\end{proof}

\begin{corollary}[Optimality Sandwich]\label{cor:optimality-sandwich}
Within the fixed-path nearest-neighbor $K$-bus routing model, the routing-round complexity of a requested layer with conflict graph $\Gamma$ satisfies
\[
\left\lceil \frac{\chi(\Gamma)}{K} \right\rceil
\leq
R_K(\Gamma)
\leq
\left\lceil \frac{\chi(\Gamma)}{K} \right\rceil.
\]
Hence
\[
R_K(\Gamma)=\left\lceil \frac{\chi(\Gamma)}{K} \right\rceil,
\]
so the routing-round complexity is exactly characterized by the chromatic number of the route conflict graph.
\end{corollary}

\begin{proof}
The lower bound and matching upper bound are exactly the two halves of the proof of Theorem~\ref{thm:congestion-round-complexity}.
\end{proof}

\begin{corollary}[Hotspot Clique Optimality]\label{cor:hotspot-clique-optimality}
If a requested routing layer contains $m$ routes that are pairwise intersecting, then its conflict graph is the clique $K_m$ and the exact minimum number of logical routing rounds in the $K$-bus model is $R_K(K_m)=\left\lceil \frac{m}{K} \right\rceil$.
\end{corollary}

\begin{proof}
The chromatic number of the clique $K_m$ is $m$, so the claim is the special case of Theorem~\ref{thm:congestion-round-complexity}.
\end{proof}

\begin{corollary}[Qubit-Only Obstruction and Ancilla Tradeoff]\label{cor:qubit-only-obstruction}
Any exact qubit-only implementation that preserves one resident logical qubit at a node and simultaneously supports $K$ independently recoverable routed bits at that same node requires at least $K$ additional colocated qubits there. Equivalently, $q$ colocated qubits can support at most $q-1$ such routed bits. In particular, a one-qubit-per-node architecture cannot realize exact same-node overlap routing at all.

More generally, if every congested node in a qubit architecture provides at most $q$ colocated qubits in total, then any requested routing layer with conflict graph $\Gamma$ requires at least
$\left\lceil \frac{\chi(\Gamma)}{q-1} \right\rceil$ logical routing rounds.
\end{corollary}

\begin{proof}
By Proposition~\ref{prop:state-count-lower-bound}, storing one resident logical bit together with $K$ independently recoverable routed bits requires local dimension at least $2^{K+1}$. A bank of $q$ colocated qubits has local dimension $2^q$, so exact same-node overlap routing is possible only if $2^q \geq 2^{K+1}$, equivalently $q \geq K+1$. Thus $K$ routed bits require at least $K$ additional local qubits beyond the resident logical qubit, and any fixed budget of $q$ colocated qubits supports at most $q-1$ routed bits.

Applying Theorem~\ref{thm:congestion-round-complexity} with effective per-node transit capacity $K=q-1$ gives the round lower bound $\left\lceil \frac{\chi(\Gamma)}{q-1} \right\rceil$.
For $q=1$, no positive transit capacity exists, so exact same-node overlap routing is impossible.
\end{proof}

\begin{remark}
By Corollaries~\ref{cor:hotspot-clique-optimality} and~\ref{cor:qubit-only-obstruction}, a one-qubit-per-node architecture cannot realize even the case of two overlapping routed controls at one hotspot without leaving the exact same-node overlap model, while a qudit with $d \geq 2^{K+1}$ supports $K$ such overlapping routes in one logical round. This is the clean congestion-level separation delivered by spectral routing: it replaces forced serialization at a hotspot by $K$-way local multiplexing inside one physical site.
\end{remark}

\begin{corollary}[Qubit--Qudit Separation]\label{cor:qubit-qudit-separation}
Under the exact nearest-neighbor unitary model introduced in Sec.~\ref{sec:formalism}, there exists a routing instance realizable with one qudit per node that is not realizable with one qubit per node.
\end{corollary}

\begin{proof}
Consider two routed operations whose chosen paths intersect at a common intermediate node and whose resident local state at that node must be preserved. In the qudit architecture with local dimension $d \geq 2^{2+1}=8$,
Theorem~\ref{thm:congestion-round-complexity} together with Corollary~\ref{cor:hotspot-clique-optimality} shows that this hotspot instance is realizable in one logical round using $K=2$ buses. By Theorem~\ref{thm:single-qubit-impossibility}, one qubit at the same node cannot simultaneously preserve an arbitrary resident qubit state and carry even one independently recoverable routed bit, hence cannot realize the two-route overlap instance at all. Therefore the instance is realizable with one qudit per node but not with one qubit per node.
\end{proof}

\subsection{Comparison with the SWAP Baseline}

\begin{corollary}
Standard SWAP-based routing is not optimal at the logical-routing level considered here, because it incurs a larger leading transport-depth constant than the proposed protocol.
\end{corollary}

\begin{proof}
Naive SWAP-based routing requires $L$ swaps to move a qubit across distance $L$, with each swap decomposing into 3 CNOT gates. Thus the transport cost is $3L$ CNOT-depth layers, compared to $2L+1$ logical routing primitives in the proposed method. Under the compilation model of Sec.~\ref{sec:formalism}, this establishes a smaller leading routing constant even though the precise device-level pulse costs remain hardware dependent.
\end{proof}

\section{Benchmark Details}\label{app:benchmark-details}

This appendix gathers the numerical material moved out of the main text. None of the benchmark content is removed; it is simply separated from the main line of the theoretical development.
Code for the experiments can be obtained from this Github Repository \cite{github2026}.
\subsection{Detailed Fan-In Sanity Checks}

To sanity-check the algebraic extension, we implemented small instances of the routed Boolean fan-in gate in both an ideal state-vector model (Cirq) and a multilevel open-system model (QuTiP). These computations are not intended as a broad hardware-performance claim; rather, they test whether explicit routed circuits behave as predicted by the formalism on representative small examples.

In the ideal Cirq model, we tested a three-control routed conjunction $g(x_1,x_2,x_3)=x_1x_2x_3$ with local dimension $d=16$ and target unitary $U=X$. An exhaustive truth-table check over all $2^3$ control strings and both target basis states produced worst-case infidelity $0$ and residual routing-subspace population $0$ after cleanup. A coherent superposition test on the same instance matched the predicted fan-in unitary to machine precision, with global infidelity $2.2\times 10^{-16}$. Thus, at the ideal unitary level, the routed construction reproduces exactly the Boolean fan-in action claimed in Theorem~\ref{thm:swap-free-multicontrol}.

In the ideal-limit QuTiP model, we tested a two-control routed conjunction with $d=8$ for path lengths $L=1$ and $L=2$ under $T_1 = T_\phi = \infty$ and zero leakage. In both cases the final-state fidelity to the ideal routed fan-in output was $1$, with zero residual routing-subspace population on both the payload nodes and the target. Hence the multilevel open-system model agrees exactly with the algebraic construction when noise is removed as shown in (Figure~\ref{fig:qutip-fanin-plots}).

\begin{figure*}[tbp]
\centering
\begin{minipage}[t]{0.48\linewidth}
\centering
\includegraphics[width=\linewidth]{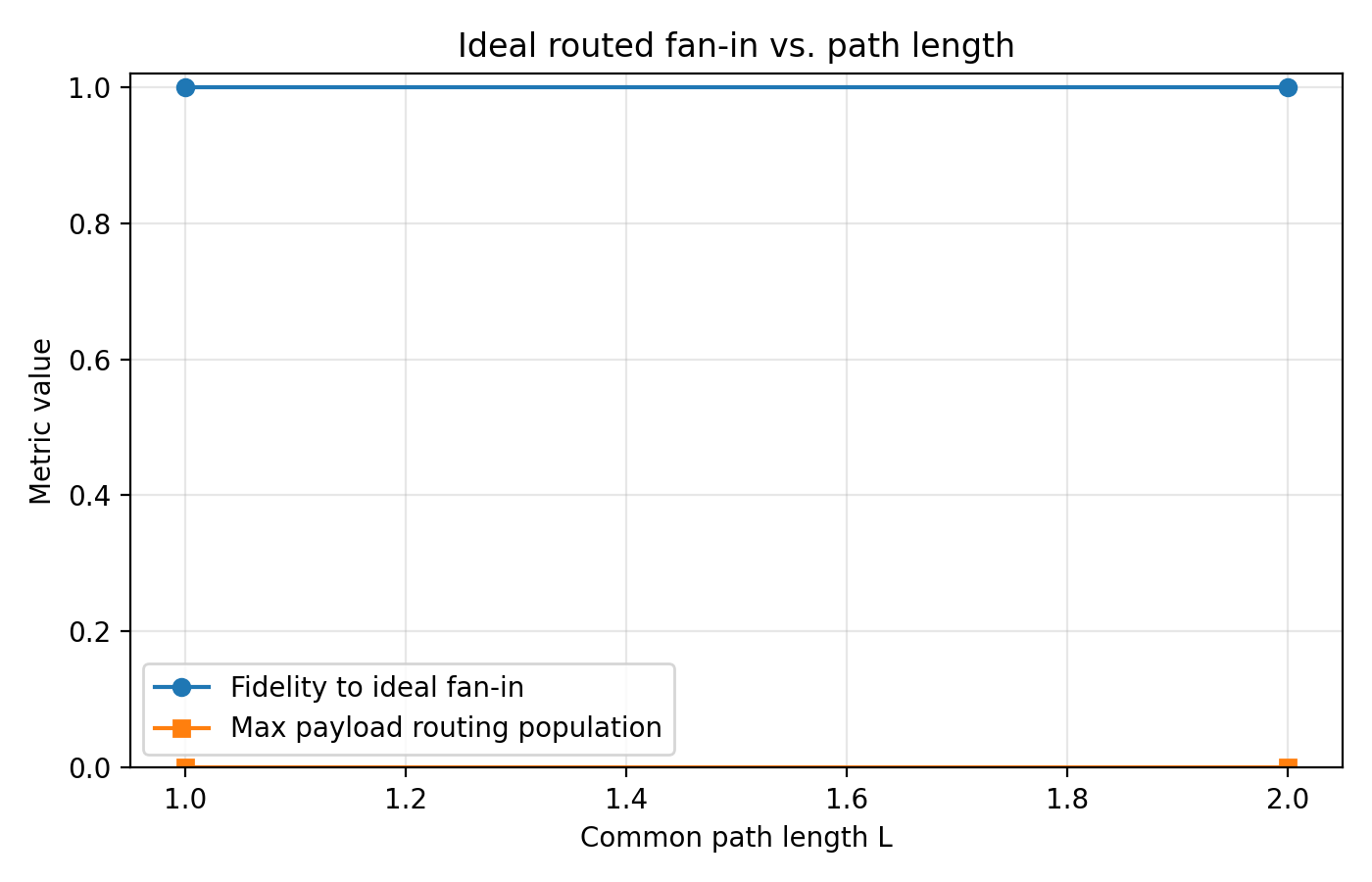}
\end{minipage}
\hfill
\begin{minipage}[t]{0.48\linewidth}
\centering
\includegraphics[width=\linewidth]{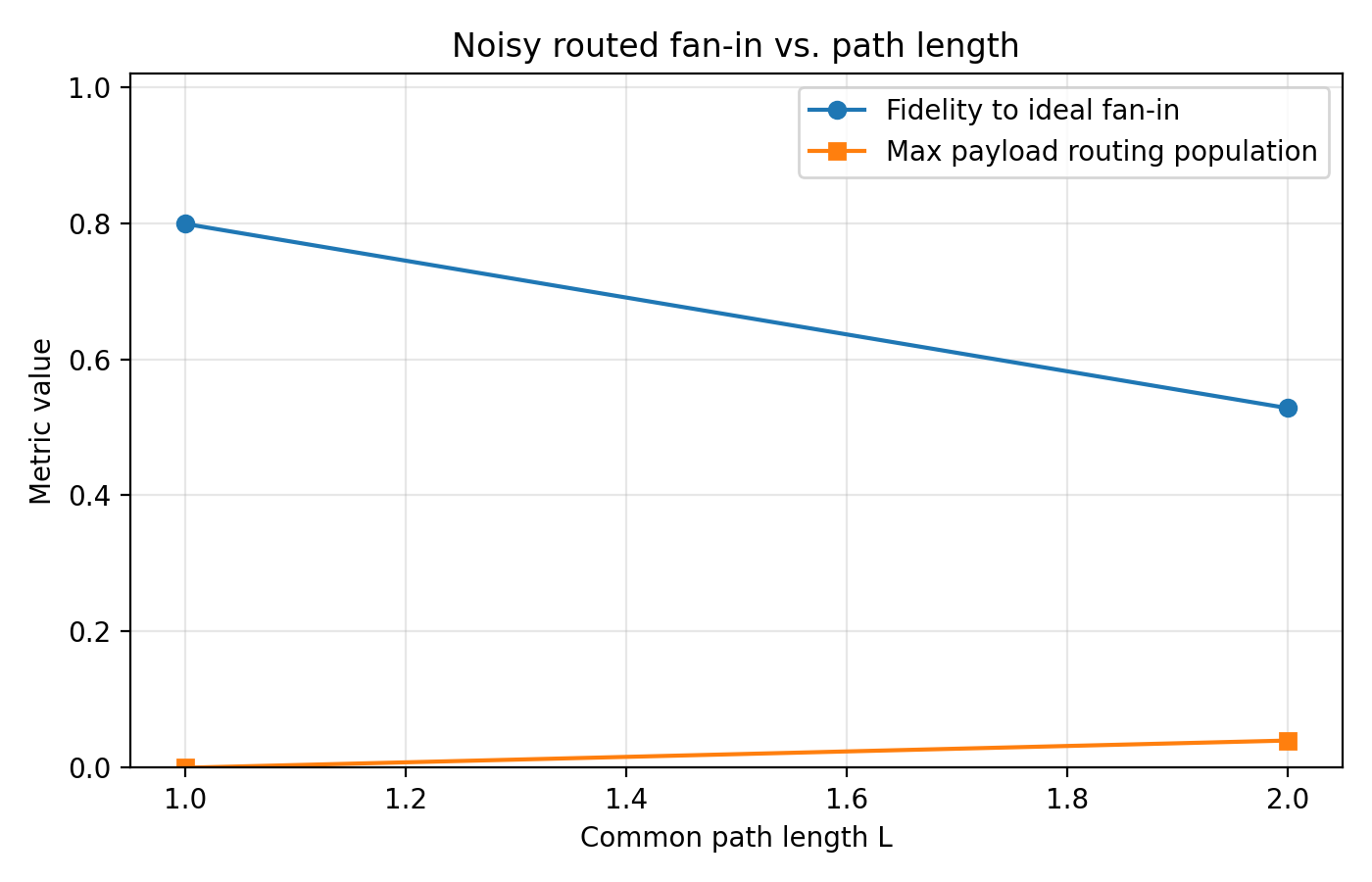}
\end{minipage}
\caption{QuTiP length-sweep plots for routed fan-in. Left: ideal-limit run with $T_1=T_\phi=\infty$ and zero leakage, showing exact agreement with the algebraic fan-in construction. Right: noisy run with level-dependent relaxation, dephasing, and leakage, showing the resulting fidelity loss as path length increases.}
\label{fig:qutip-fanin-plots}
\end{figure*}

As a first noisy stress test, we also ran a two-control QuTiP sweep with $d=8$, routing and target gate times both set to $1$, level-dependent relaxation times $(80,45,25,14)$, dephasing times $(120,80,50,30)$, leakage parameter $\epsilon=0.002$, and $50$ Monte Carlo trajectories. For routed conjunction, the fidelity to the ideal fan-in output was approximately $0.799$ at $L=1$ and $0.528$ at $L=2$. The maximal payload-node routing population was $0$ at $L=1$ and $0.04$ at $L=2$, while the target routing population was approximately $0.14$ at $L=1$ and $0.02$ at $L=2$. Thus the noisy simulation tracks the same routed fan-in circuit under this noise model, but the state fidelity degrades substantially once higher-level decay, dephasing, and leakage are introduced. These numbers should therefore be read as an illustrative small-instance stress test rather than as a converged hardware threshold study.

The same fan-in experiments also make the target-contention constant visible. In a concrete Cirq moment schedule for the three-control conjunction, the constructed routed depth scales as $2L + 5$ rather than the schematic $2L + 1$,
because the last-hop propagation gates into a shared target cannot all be executed in the same moment. This is still consistent with the asymptotic statement of Proposition~\ref{thm:depth-localization}, but it shows how the additive constant can depend on shared-target scheduling.

\begin{figure*}[tbp]
\centering
\includegraphics[width=0.72\linewidth]{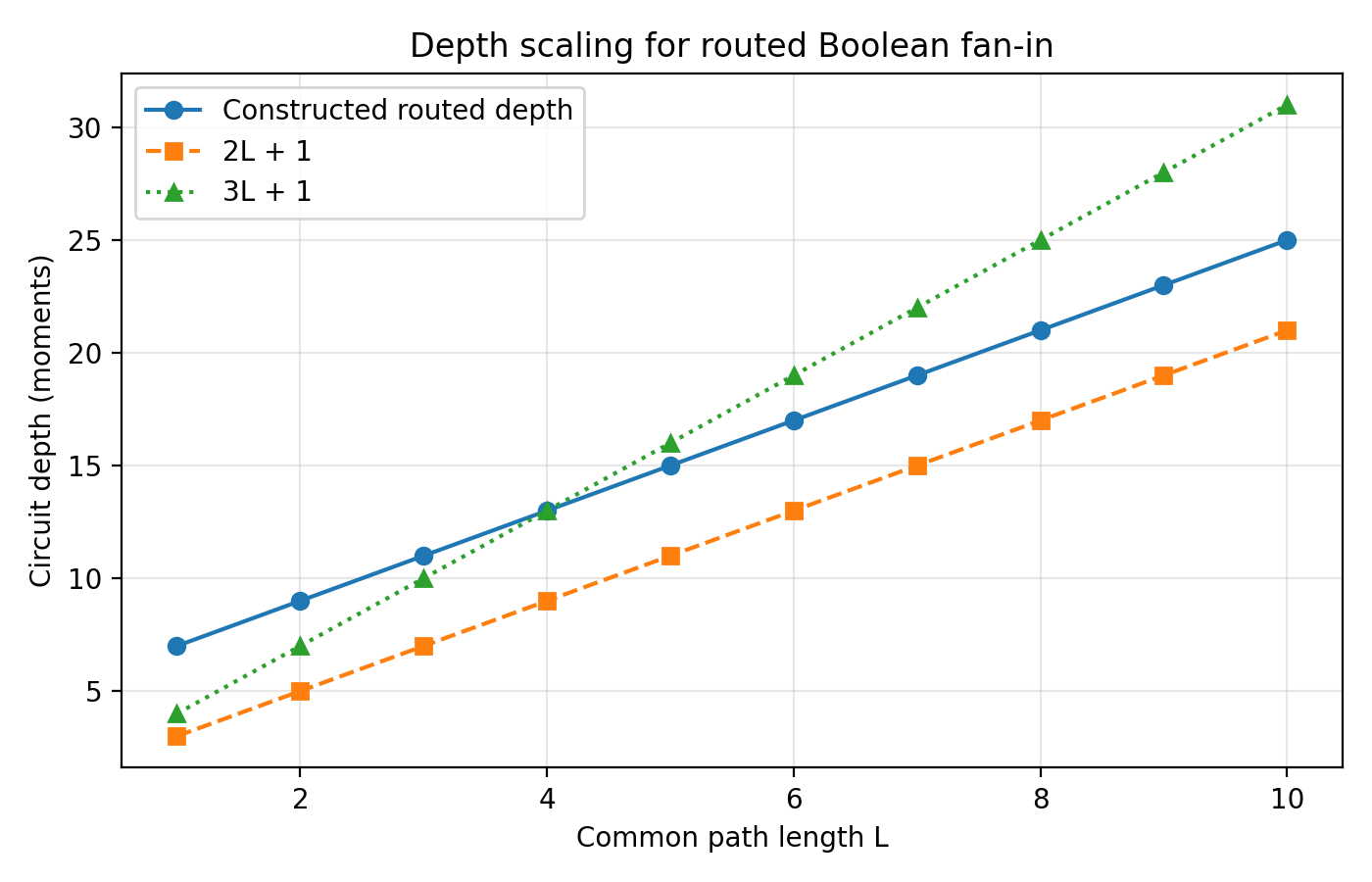}
\caption{Concrete Cirq depth scaling for the routed fan-in construction. The plotted circuit moments agree with the linear-in-$L$ behavior predicted by the theory, while also making visible the additive overhead created by shared-target scheduling contention.}
\label{fig:cirq-fanin-depth-data}
\end{figure*}

\subsection{Detailed Ideal Cirq Benchmarks}

To validate the ideal guarantees of the routing protocol on small instances, we simulated the architecture in Google Cirq using exact state-vector evolution. The verification suite implements the generalized $CBL$, $BCP$, and routing-controlled target primitives as explicit qudit unitaries and records the resulting benchmarking data in the accompanying output files(Figure~\ref{fig:cirq-fanin-depth-data}).

\paragraph{Exact vs.\ proxy regimes}
Exact Cirq validation is reported only for runs that satisfy the exact-evaluation and dimension requirements. Larger-$N$ sweeps and dimension-limited configurations are reported in metrics-only or proxy regimes, and are interpreted as structural scaling evidence rather than full-state validation.

\paragraph{State-vector feasibility}
Exact state-vector simulation scales as $d^N$ for $N$ qudits of local dimension $d$. For the small-scale benchmarks (mean $N \approx 6$), dimensions such as $d=8$ remain tractable, so exact fidelity and cleanup checks are meaningful. For the large-scale sweeps (e.g., $N \approx 30$), even modest dimensions like $d=8$ imply a Hilbert space of size $8^{30} = 2^{90}$, far beyond available classical memory. Those large-$N$ runs therefore use metrics-only or proxy evaluation by design; they test routing depth and congestion scaling rather than full-state correctness. This limitation is a classical simulation bound, not a failure of the routing protocol.

\paragraph{Algorithmic Correctness and Zero Crosstalk}
We first verified the single-path protocol of Theorem~\ref{thm:swap-free-nonlocal-cnot} on a chain of length $L = 4$ with local dimension $d = 8$. The simulated final state achieved a global infidelity of $4.44 \times 10^{-16}$ relative to the ideal non-local CNOT output, the intermediate nodes were exactly restored, and the residual routing population was numerically zero. This confirms that the protocol realizes the intended logical action while fully cleaning the routing subspace at the end of execution.

\begin{figure*}[htbp!]
    \centering
    \includegraphics[width=0.85\textwidth]{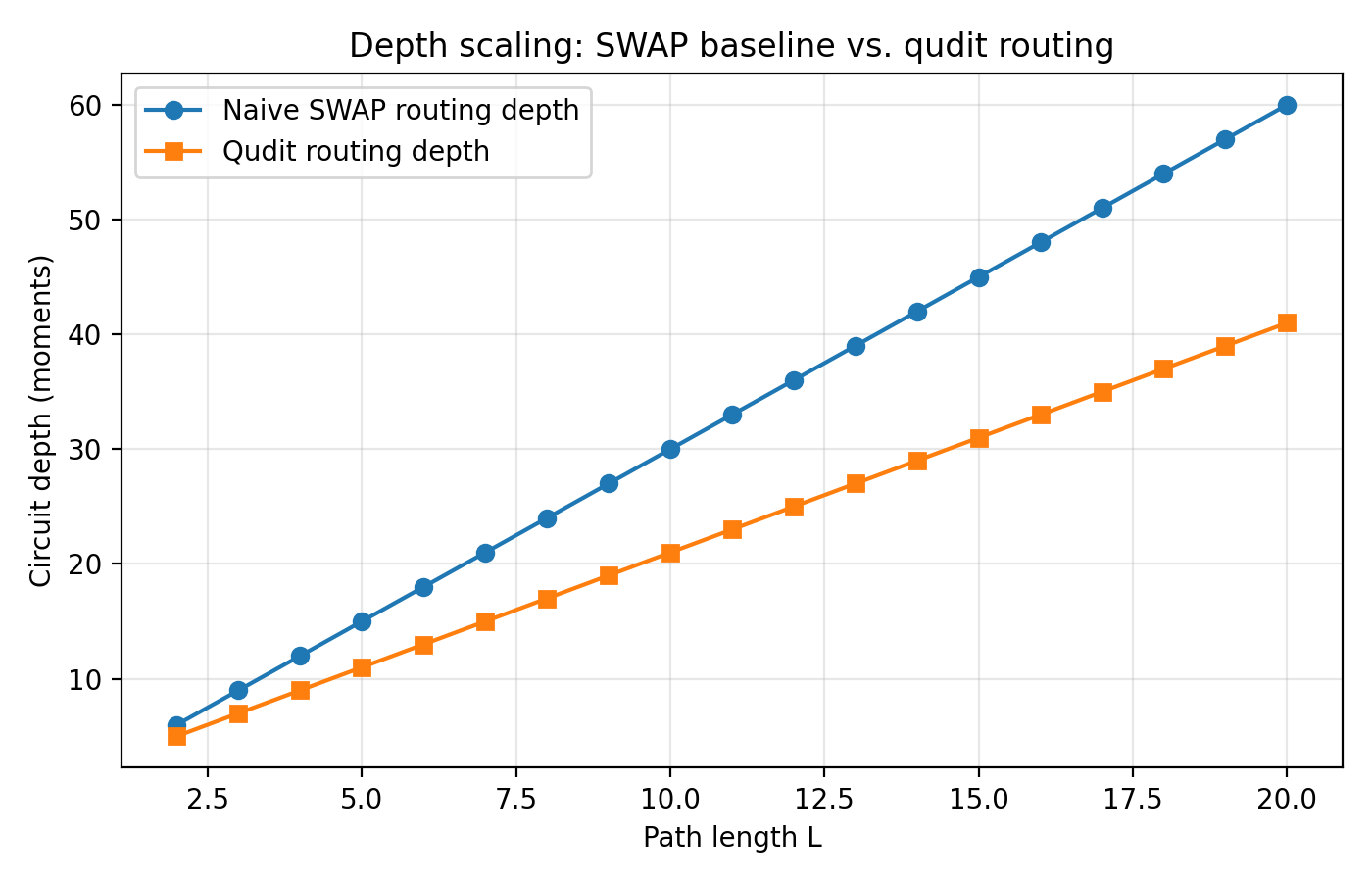}
    \caption{Compiled circuit depth for naive SWAP routing and the proposed qudit-based routing protocol as a function of path length $L$. The fitted curves recover the theoretical $3L$ versus $2L + O(1)$ scaling, with the additive constant in the qudit curve arising from the target interaction.}
    \label{fig:depth-scaling}
\end{figure*}

We next tested the overlapping-path setting predicted by Theorems~\ref{thm:bus-parallelism} and~\ref{thm:parallel-correctness} on a crossroads instance with $d = 8$, using $\Delta_1 = 2$ for the horizontal route and $\Delta_2 = 4$ for the vertical route. The overlap infidelity and final infidelity were both zero, the center-qudit level distribution matched the expected modulo-2 decomposition exactly, and the final routing population again vanished. In the ideal unitary model, this confirms that simultaneous routes sharing a physical qudit produce no logical crosstalk.

\paragraph{Depth Scaling Advantage}
We also compared the compiled circuit depth of the proposed routing protocol against a naive qubit SWAP-routing baseline over path lengths $2 \leq L \leq 20$ (Figure \ref{fig:depth-scaling}).

The fitted SWAP depth is $3L$ to numerical precision, while the qudit routing curve has slope $2$ with a constant additive offset; over the tested range the compiled depth is exactly $2L + 1$. These results numerically support Proposition~\ref{prop:protocol-depth} and Proposition~\ref{prop:swap-routing-improvement}: shifting the routing burden into higher levels preserves linear scaling while improving the leading depth coefficient relative to SWAP-based transport.

\subsection{Conflict-Graph and Compiler-Demand Validation}

We also supplemented the device-level simulations with two independent validation scripts aimed at the manuscript's congestion claims. First, an exact combinatorial validator exhaustively checked the congestion theorem on every labeled conflict graph up to five vertices. Across all $1099$ graphs and for $K=1,2,3$, the brute-force minimum number of routing rounds agreed exactly with the closed form $R_K(\Gamma)=\left\lceil \frac{\chi(\Gamma)}{K} \right\rceil$, yielding zero discrepancies. The explicit hotspot examples used in the main text likewise fall exactly on the predicted clique curve $R_K(K_m)=\lceil m/K \rceil$. The corresponding hotspot-round plot is shown below; over the plotted range the star-hotspot and line-overlap families induce the same clique conflict graphs, so the two panels are expected to coincide.
\begin{figure*}
    \centering
    \includegraphics[width=\textwidth]{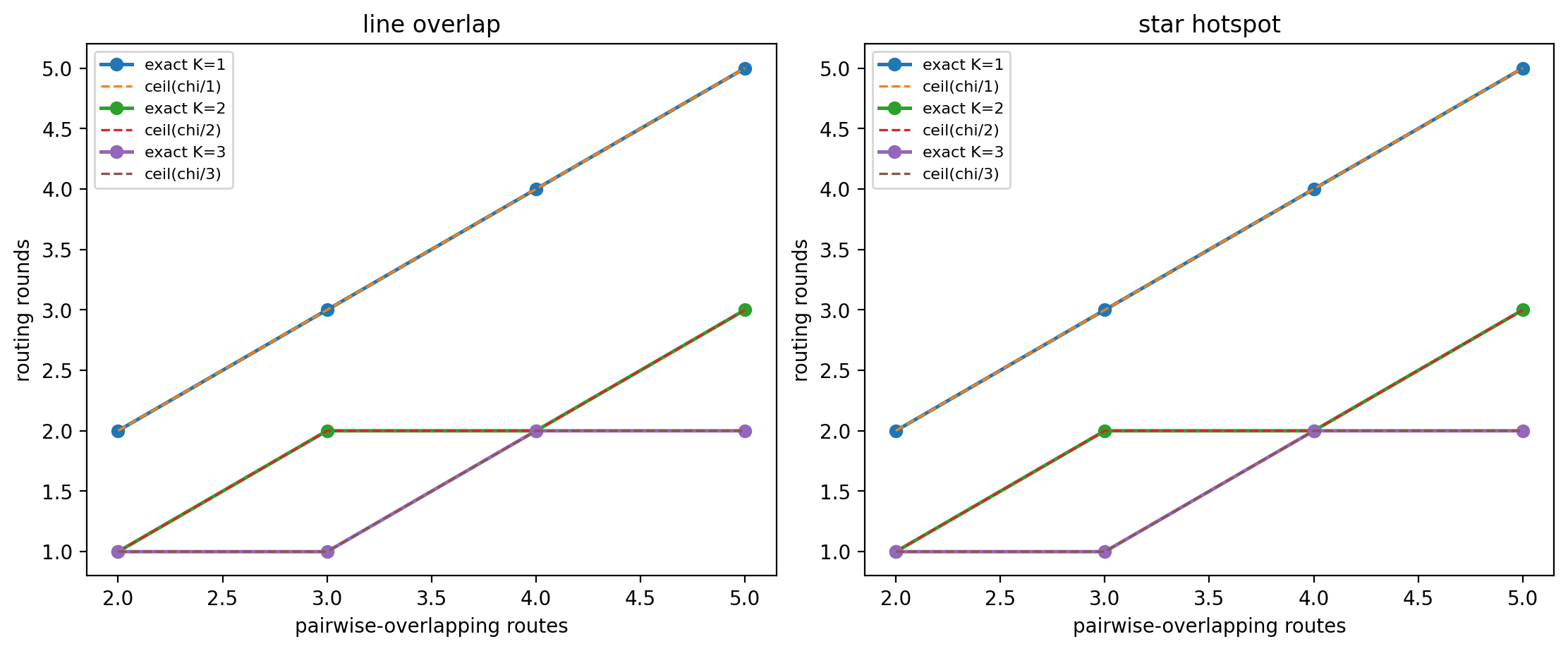}
\caption{
Validation of the congestion-round theorem on representative conflict-graph families. 
Left: line-overlap routes. Right: star-hotspot routes. 
For each instance, the minimum number of routing rounds obtained by exhaustive search matches the theoretical prediction $R_K(\Gamma) = \lceil \chi(\Gamma)/K \rceil$. 
Over the plotted range, both families induce identical clique conflict graphs, confirming that routing congestion is governed by graph coloring rather than geometric layout.
}    \label{fig:congestion-validation}
\end{figure*}
% \centering
% \includegraphics[width=\columnwidth]{images/hotspot_rounds_validation.png}
% \caption{\todo[inline]{Caption for this image}}
% \end{figure}

Second, a compiler-demand suite evaluated QFT, regular-graph QAOA, Grover-style amplitude amplification, and mirror-interaction workloads on line and grid topologies using only shortest-path routing with uniform tie sampling at a fixed~\ref{fig:congestion-validation}. On the representative $n=8$ cases summarized in Table~\ref{tab:compiler-route-benchmarks}, the routed-to-SWAP transport ratios range from $0.75$ to $0.89$. In the direct LNN overlap with recent SWAP-less QFT/QAOA benchmarks, the line-topology ratios are $0.778$ for QFT and $0.784$ for QAOA. The shared QFT and QAOA application layer therefore permits a direct application-family comparison, although our present benchmark remains a route-demand comparison rather than a full gate-library compilation comparison. Thus the conflict-graph picture is not confined to hand-drawn path examples: it persists on explicit nontrivial compiled workloads under a methodology that is deliberately constrained not to favor the routed model by custom path engineering.

\paragraph{Appendix LNN application summary.}
Unlike the main-text benchmark table, the compact summary below isolates only the linear-chain cases, which are the direct geometry overlap with recent swapless application studies and therefore the cleanest place to compare application families without repeating the full line-and-grid dataset.
\begin{table}[ht]
    \centering
    \scriptsize
 \resizebox{\columnwidth}{!}{%
    \begin{tabular}{l c c c}
    \hline
    Family & mean SWAP transport & mean routed transport & mean $R_2$ \\
    \hline
    QFT line (LNN) & 19.38 & 15.08 & 1.38 \\
    QAOA line (LNN) & 15.69 & 12.31 & 1.00 \\
    Amplitude line (LNN) & 7.42 & 6.06 & 1.00 \\
    Mirror line (LNN) & 48.00 & 36.00 & 2.00 \\
    \hline
    \end{tabular}}
    \caption{Comparison between SWAP-based routing and our work}
    \label{tab:swapcompare}
\end{table}
% \begin{center}
% \scriptsize
% \resizebox{\columnwidth}{!}{%
% \begin{tabular}{l c c c}
% \hline
% Family & mean SWAP transport & mean routed transport & mean $R_2$ \\
% \hline
% QFT line (LNN) & 19.38 & 15.08 & 1.38 \\
% QAOA line (LNN) & 15.69 & 12.31 & 1.00 \\
% Amplitude line (LNN) & 7.42 & 6.06 & 1.00 \\
% Mirror line (LNN) & 48.00 & 36.00 & 2.00 \\
% \hline
% \end{tabular}
% }
% \end{center}

\subsection{Detailed Noisy QuTiP Benchmarks}

To probe the physical limits of the routing idea beyond the ideal Cirq checks of Sec.~\ref{sec:numerical-benchmarking}, we simulated open-system dynamics in QuTiP using the Lindblad master equation. The purpose of these noisy simulations is not to demonstrate immediate superiority over SWAP-based routing, but to identify the hardware regime in which the architectural advantages of spectral routing become operational. The verification code supports level-dependent $T_1$ relaxation, level-resolved pure dephasing times $T_\phi$, and optional local cross-level leakage channels after routing primitives. For the larger Hilbert-space instances in our HPC workflow, the code switches from exact master-equation evolution to a Monte Carlo wave-function treatment using QuTiP's trajectory solver; for the final archived HPC runs, every such dataset was generated with seed $0$ and $500$ trajectories. The specific distance-sweep datasets reported below isolate the amplitude-damping contribution by setting $T_\phi = \infty$ and leakage $\epsilon = 0$, so the reported curves should be read as deliberately stripped-down $T_1$-limited benchmarks rather than a full noise budget. The corresponding distance-sweep plot is shown in Fugure~\ref{fig:distance-sweep}.
\begin{figure}[ht!]
    \centering
    \includegraphics[width=\columnwidth]{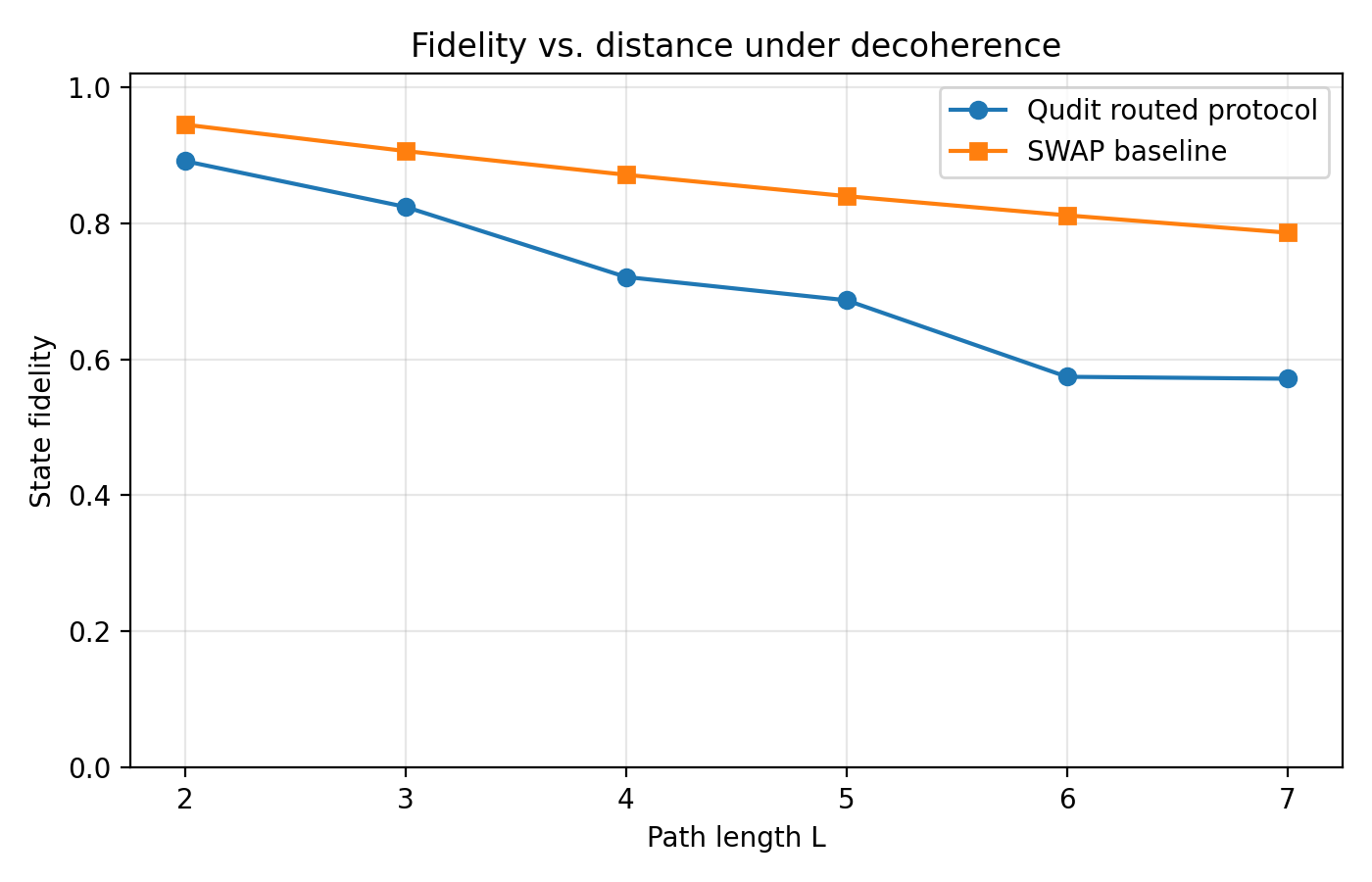}
    \caption{
State fidelity as a function of path length $L$ under decoherence, comparing the routed protocol (blue) with a SWAP-based baseline (orange). 
The routed implementation exhibits faster fidelity decay in this regime due to the involvement of higher-energy levels, illustrating that the architectural advantage of spectral routing depends on sufficient coherence and control quality of the additional qudit levels. 
The plot highlights the tradeoff between reduced routing depth and increased noise sensitivity.
}
    \label{fig:distance-sweep}
\end{figure}

The final seeded exploratory $d=5$ sweep reinforces the same conclusion. There, the routed fidelity decreases from $0.8909$ at $L=2$ to $0.8155$ at $L=3$, $0.7207$ at $L=4$, $0.6665$ at $L=5$, and $0.5797$ at $L=6$, while the SWAP baseline remains above it throughout. In particular, the late-distance rebound visible in an earlier unseeded exploratory artifact disappears in the final seeded $500$-trajectory run, so we no longer treat that rebound as evidence of a physical asymptote in the observed data. The final archived QuTiP results therefore do \emph{not} support a crossover at $L=3$, nor do they justify the stronger statement that the routed protocol dominates once paths become moderately long. What they do show is the relevant engineering bottleneck: if higher-level coherence is poor, the spectral-routing advantage remains mathematically correct but physically unrealized. The corresponding threshold scan is shown below.

\begin{figure}[tbp]
\centering
\includegraphics[width=\columnwidth]{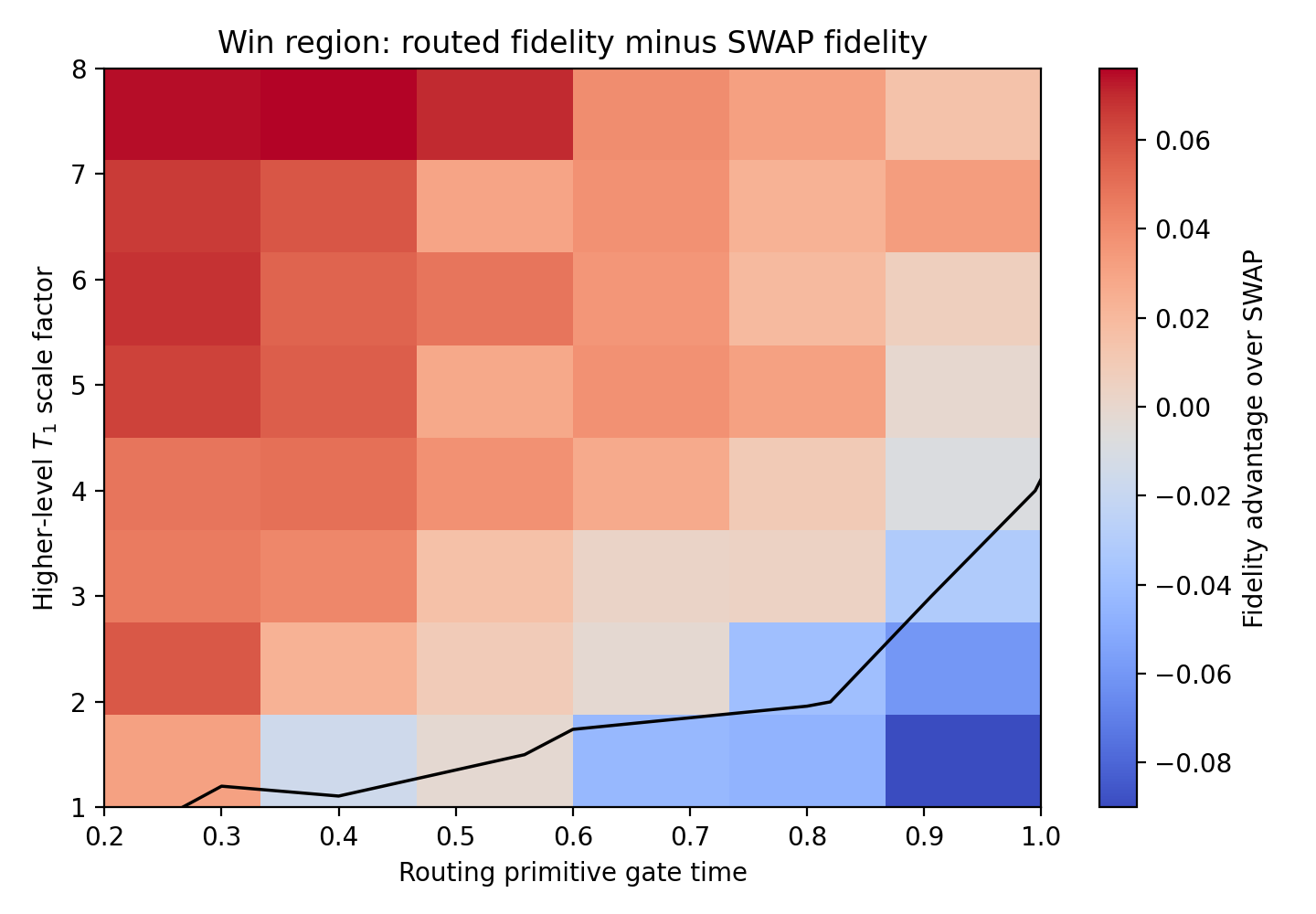}
\caption{
    Win region for the routed protocol under decoherence, shown as routed fidelity minus SWAP fidelity across a sweep of routing primitive duration and higher-level coherence parameters. 
    Positive values (red) indicate regimes where spectral routing outperforms the SWAP baseline, while negative values (blue) indicate SWAP superiority. 
    The plot identifies a threshold regime in which improved coherence and/or faster routed primitives enable a crossover to routed advantage, illustrating that the benefits of spectral routing are conditional on physical device parameters.
}
\label{fig:win-region}
\end{figure}

At the same time, the final threshold scan shows that the physical conclusion is conditional rather than negative in principle. Fixing $L=3$, $d=8$, seed $0$, and $500$ Monte Carlo trajectories, we scanned the routed primitive duration together with a multiplicative improvement applied only to the higher-level lifetimes $|2\rangle, |3\rangle, \dots$. The resulting win region identifies the hardware regime in which the routed protocol becomes fidelity-superior to SWAP on the sampled grid. In this coarse scan, the crossover appears near a factor of $5$ when the routed primitive duration is $1.0$; near a factor of $2$ when the routed primitive duration is $0.8$ or $0.6$; near a factor of $1.5$ when the routed primitive duration is $0.4$ or $0.3$; and already at the baseline higher-level lifetime point when the routed primitive duration is $0.2$. These values should be read as approximate crossover estimates from a single-seed Monte Carlo study with $T_\phi = \infty$ and zero leakage, not as sharp universal thresholds. Even so, the QuTiP study does more than say that the present hardware model is unfavorable: it maps the coherence-speed regime in which the theoretical depth advantage is likely to become a practical fidelity advantage.

\paragraph{Note on Asymptotic Relaxation}
If the routed circuit is pushed to sufficiently extended distances, its fidelity should asymptotically approach $0.5$ once $T_1$ amplitude damping fully dominates the total evolution time. In that regime, residual population in the routing subspace relaxes to the global ground state $|00\dots0\rangle$. Because the target logical state is Bell-like, the fully relaxed ground state still retains a $0.5$ overlap with the ideal target. This limiting value therefore reflects total thermodynamic relaxation, not a recovery of coherent routing. The final seeded HPC data shown here does not yet reach that asymptotic regime directly, but it is consistent with the same physical interpretation of sufficiently long-time relaxation.

\end{document}